\documentclass{scrartcl}
\pdfoutput=1
\usepackage{amsmath}
\usepackage{amsthm}
\usepackage{amssymb}
\usepackage{float}
\usepackage[algoruled,linesnumbered,noend]{algorithm2e}
\usepackage{relsize}

\SetCommentSty{mycommentsty}
\usepackage{cite}
\usepackage{multirow}
\usepackage{url}
\usepackage{xspace}
\usepackage{graphicx}
\usepackage{microtype}

\usepackage{verbatim}
\usepackage{textcomp}\usepackage{xcolor}

\usepackage{tikz}
\usetikzlibrary{shapes,arrows}
\usetikzlibrary{decorations}
\usetikzlibrary{positioning,patterns}

\usepackage{booktabs}

\usepackage{tabularx}
\newcolumntype{W}{>{\raggedleft\arraybackslash}X}
\newcolumntype{C}{>{\centering\arraybackslash}X}
\usepackage{array}

\newcommand{\ie}{\textit{i.e.},\xspace}

\newcommand{\wrt}{w.r.t.\xspace}

\renewcommand{\l}{\ensuremath{\ell}}
\renewcommand{\emptyset}{\ensuremath{\varnothing}}
\newcommand{\rev}{\ensuremath{\mathit{r}}}
\newcommand{\SA}{\ensuremath{\mathit{SA}}}
\newcommand{\LCP}{\ensuremath{\mathit{L}}}
\newcommand{\Occ}{\ensuremath{\mathit{Occ}}}
\newcommand{\FileBWT}{\ensuremath{\mathcal{B}}\xspace}
\newcommand{\FileLCP}{\ensuremath{\mathcal{L}}\xspace}
\newcommand{\FileSA}{\ensuremath{\mathcal{SA}}\xspace}

\newcommand{\FileE}{\ensuremath{\mathcal{E}}}
\newcommand{\FileP}{\ensuremath{\mathcal{P}}}

\newcommand{\FileArc}{\ensuremath{\mathcal{A}}}

\newcommand{\extalphabet}{\Sigma^{\$}}

\newtheorem{theorem}{Theorem}
\newtheorem{lemma}[theorem]{Lemma}
\newtheorem{proposition}[theorem]{Proposition}

\newtheorem{corollary}[theorem]{Corollary}

\theoremstyle{definition}
\newtheorem{definition}{Definition}

\begin{document}

\title{An External-Memory Algorithm for String Graph Construction}
\author{Paola Bonizzoni \and  Gianluca Della Vedova \and  Yuri
  Pirola \and  Marco Previtali \and   Raffaella Rizzi}

\date{DISCo, Univ. Milano-Bicocca, Milan, Italy}

\maketitle

\begin{abstract}
Some recent results~\cite{Rosone2013,Cox2012} have introduced
external-memory algorithms to compute self-indexes of a set of
strings, mainly via computing the Burrows-Wheeler Transform (BWT) of
the input strings.
The motivations for those results stem from Bioinformatics, where a
large number of short strings (called reads) are routinely produced
and analyzed.
In that field,
a fundamental problem is to assemble a genome from a large set of much
shorter samples extracted from the unknown genome.
The
approaches that are currently used to tackle this problem are
memory-intensive.
This fact does not bode well with the ongoing increase in the
availability of genomic data.
A data structure that is used in genome assembly is the string graph,
where vertices correspond to samples and arcs represent two
overlapping samples.
In this paper we address an open problem of~\cite{Simpson2010}: to
design an external-memory algorithm to compute the string graph.
\end{abstract}

\section{Introduction}

Several fields are witnessing a huge increase in the amount of available data,
such as real-time network data, Web usage logs, telephone call records,
financial transactions, and biological data~\cite{chen2002multi,Benson01012014}.
There are three main algorithmic solutions to cope with that amount of data: (1)
data streaming, where only one pass is made over the data and the working memory
is small compared to the input
data~\cite{demetrescu_trading_2009,computing-data-streams-Henzinger-1999}, (2)
parallel algorithms, where input data is split among several
processors~\cite{Valiant:1991:GPP:114872.114890}, and (3)
external memory algorithms~\cite{DBLP:journals/cacm/AggarwalV88,Vitter:2001:EMA:384192.384193} where only part
of the data is kept in main memory and most of the data are on disk.

The latter subfield has blossomed with the seminal paper by Vitter and
Shriver~\cite{vitter-shriver-1}, introducing the \emph{parallel disk
model}, where the performance is measured as the number of I/O operations and
the amount of disk space used.

The above discussion is especially relevant in
Bioinformatics, where we are currently witnessing a tremendous increase in
the data available, mainly thanks to the rise of different Next-Generation Sequencing
(NGS) technologies~\cite{PMID:20644199}.
De novo sequence assembly is still one of the most fundamental problems and is currently
receiving a lot of attention, just as it used to be twenty years ago~\cite{alizadeh_physical_1995,alizadeh_physical_1995-1}.
The assembly problem asks for a superstring $G$ (the unknown genome)
of the set $R$ of input strings (sampled from the unknown genome).
The concatenation of all input strings is a feasible solution of the problem,
but it is clearly a solution void of any biological significance: for this
reason a suitable optimization criterion must be introduced.
The simplest criterion is to find a shortest superstring of $R$~\cite{BJLTY94},
but that model considers neither that the input strings are sampled uniformly
from the genome $G$, nor that the samples may contain some errors (that is $G$
is a superstring only in an approximate sense).
Moreover, data obtained with different technologies or different instruments
have different characteristics, such as the length of the samples and the error
distribution, making difficult to describe a unified computational problem that
actually represents the real-world genome assembly problem.

For all those reasons, the successful assemblers incorporate a number of ideas
and heuristics originating from the biological characteristics of the input data
and of the expected output.
Interestingly, almost all the assemblers are based on some notion of graph to
construct a draft assembly.
Most of the available assemblers~\cite{Simpson2009,Peng2010,bankevich2012spades}
are built upon the notion of de Bruijn graphs, where the vertices are all
distinct $k$-mers (that is the $k$-long substrings appearing in at
least an  input string).
If we want to analyze datasets coming from different technologies, hence with
important variability in read lengths, an approach based on same-length strings is
likely to be limiting.
Moreover, one of the main hurdles to overcome is the main memory space that is
used by those assemblers.
For instance, a standard representation of the de Bruijn graph for the human
genome when $k=27$ requires 15GB (and is unfeasible in metagenomics).
To reduce the memory usage,  a probabilistic version of de
Brujin graphs, based on the notion of Bloom
filter, has been introduced~\cite{DBLP:journals/almob/ChikhiR13} and
uses less than 4GB of memory to store the de Bruijn graph for the human
genome when $k=27$.

The amount of data necessary to assemble a genome emphasizes the need for
algorithmic solutions that are time and space efficient.
An important challenge is to reduce main memory usage while keeping a reasonable
time efficiency.
For this reason, some
alternative approaches have been developed recently, mostly based on
the idea of \emph{string graph}, initially proposed by Myers~\cite{Myers2005}
before the advent of NGS technologies, and further
developed~\cite{Simpson2010,Simpson2012} to incorporate some advances in text
indexing, such as the FM-index~\cite{Ferragina2005}.
These methods build a graph whose vertices are the input reads and a visit of
the  paths of the graph allows to reconstruct the genome.

A practical advantage of string graphs over de Bruijn graphs is that reads are
usually much longer than $k$, therefore string graphs can immediately disambiguate some repeats
that de Bruijn methods might resolve only at later stages.
On the other hand, string graphs are more computationally intensive to compute~\cite{Simpson2012}.
For this reason we have studied the problem of computing the string graph on a
set $R$ of input strings, with the goal of developing an external-memory
algorithm that uses only a limited amount of main memory, while minimizing disk accesses.

Our work has been partially inspired by SGA~\cite{Simpson2010}, the most used string graph assembler.
SGA, just as several other bioinformatics programs, is based on the notions of
BWT~\cite{Burrows1994} and of FM-index constructed from the set $R$ of reads.
In fact, an important distinguishing feature of SGA is its use of the FM-index to compute the arcs of the string graph.
Still, the memory usage of SGA is considerable, since the experimental analysis
in~\cite{Simpson2012} has proved that SGA can successfully
assemble the human genome from a set of $\approx$1 billion 101bp reads, but
uses more than 50GB of RAM to complete the task.
The space improvement achieved in  the latest SGA implementation~\cite{Simpson2012}  required   to apply  a distributed construction algorithm of
the  FM-index for the collection of reads and specific optimization strategies to avoid  keeping the whole indexing of reads in main memory.
Indeed, the authors of SGA~\cite{Simpson2010}
estimated 400GB  of main  memory to  build the FM-index for a
collection of reads at 30x coverage over the human genome.
Since the approach used in SGA~\cite{Simpson2010}  requires to keep in main
memory the entire BWT and the FM-index of all input data, an open problem
of~\cite{Simpson2010} is to  reduce the space requirements by developing an
external memory algorithm to compute the string graph.
In this paper we are going to address this open problem.

Another fundamental inspiration is the sequence of
papers~\cite{Rosone2013,cox_large-scale_2012,bauer_lightweight_2013} that have
culminated in BCRext~\cite{bauer_lightweight_2013}, a lightweight (\ie external-memory) algorithm to compute the
BWT (as well as a number of other data structures) of a set of strings.
In fact, our algorithm receives as input all data structures computed by BCRext
and the set $R$ of input reads,  computing the string graph of $R$.

\section{Definitions}

We  briefly recall some standard definitions that will be used in the following.
Let $\Sigma$ be an ordered finite
alphabet and let $S$ be a string over $\Sigma$.
We denote by $S[i]$ the $i$-th symbol of $S$, by $\l=|S|$ the length of $S$, and by
$S[i:j]$ the substring $S[i]S[i+1] \cdots  S[j]$ of $S$.
The \emph{reverse} of $S$ is the string $S^{\rev}=S[\l] S[\l-1] \cdots S[1]$.
The \emph{suffix} and \emph{prefix} of $S$ of length $k$ are the
substrings $S[\l-k +1: \l]$ (denoted by $S[\l-k +1:]$) and $S[1: k]$ (denoted by
$S[:k]$) respectively.
The $k$-suffix of $S$ is the $k$-long suffix of $S$.

Given two strings $(S_i, S_j)$, we say that $S_i$
\emph{overlaps} $S_j$  iff a nonempty suffix $Z$ of $S_i$ is
also a prefix of $S_j$, that is $S_{i}=XZ$ and $S_j = Z Y$.
In that case we say that that $Z$ is the \emph{overlap} of $S_i$ and $S_j$,  denoted as $ov_{i,j}$,
that $Y$ is the \emph{right extension} of $S_i$ with $S_j$, denoted as
$rx_{i,j}$, and $X$ is
the  \emph{left extension}  of  $S_j$ with $S_i$, denoted as
$lx_{i,j}$.

In the following of the paper we will consider a collection
$R = \{r_1, \dots, r_m\}$ of $m$ strings (also called \emph{reads}, borrowing the term
from the bioinformatics literature)  over $\Sigma$.
As usual, we append a sentinel symbol $\$\notin\Sigma$ to the end of
each string ($\$$ lexicographically precedes all symbols in $\Sigma$) and  we
denote by $\extalphabet$ the extended alphabet $\Sigma \cup \{\$\}$.
We assume that the sentinel symbol $\$$ is not taken into account when
computing overlaps between two strings or when considering the length of a
string.
Moreover, two sentinel symbols do not match when compared with each
other.
This fact implies that two strings consisting only of a sentinel
symbol have a longest common prefix that has length zero.
A technical difficulty that we will overcome is to identify each read
even using a single sentinel.
Using a distinct sentinel for each read would eliminate the problem,
but it would make the alphabet size too large.
We denote by $n$ the total number of characters in the input strings, and by $l$
the maximum length of a string, that is $n=\sum_{i=1}^{m} |r_{i}|$ and $l =
\max_{i=1}^{m} \{ |r_{i}| \}$.

\begin{definition}
\label{definition:GSA-LCP-BWT}
The \emph{Generalized Suffix Array (GSA)}~\cite{Shi1996} of $R$ is the array
$\SA$ where each element $\SA[i]$ is
equal to $(k, j)$ if and only if the $k$-suffix of string $r_{j}$ is the
$i$-th smallest element in the lexicographic ordered set of all suffixes of the
strings in  $R$.
The \emph{Longest Common Prefix (LCP) array} of  $R$, is the $n$-long array
$\LCP$  such that $\LCP[i]$  is equal to
the length of the longest prefix shared by the the $k_{i}$-suffix of
$r_{j_{i}}$ and the  $k_{i-1}$-suffix of $r_{j_{i-1}}$, where $\SA[i]=(k_{i}, j_{i})$ and  $\SA[i-1]=(k_{i-1}, j_{i-1})$.
Conventionally, $\LCP[1]=-1$.
The \emph{Burrows-Wheeler Transform (BWT)} of $R$ is the sequence $B$ such that
$B[i]=r_{j}[k-1]$, if $\SA[i] = (k,j)$ and $k > 1$, or $B[i]= \$$, otherwise.
Informally, $B[i]$ is the symbol that precedes the $k$-suffix of
string $r_j$ where such suffix is the $i$-th smallest suffix in the
ordering given by $\SA$.
\end{definition}

\begin{table}[htb]
\centering
\begin{tabular}{r|r|r|r|c}
$i$  &$LS[i]$&\SA[i]&\LCP[i]&B[i]\\\hline
1 &\$    &$(0, 1)$& -&E\\
2 &\$  &$(0, 3)$& 0&T\\
3 &\$    &$(0, 2)$& 0&N\\
4 &APPLE\$    &$(5, 1)$& 0&\$\\
5 &APRICOT\$  &$(7, 3)$& 2&\$\\
6 &COT\$      &$(3, 3)$& 0&I\\
7 &E\$        &$(1, 1)$& 0&L\\
8 &EMON\$     &$(4, 2)$& 1&L\\
9 &ICOT\$     &$(4, 3)$& 0&R\\
10&LE\$       &$(2, 1)$& 0&P\\
11&LEMON\$    &$(5, 2)$& 2&\$\\
12&MON\$      &$(3, 2)$& 0&E\\
13&N\$        &$(1, 2)$& 0&O\\
14&ON\$       &$(2, 2)$& 0&M\\
15&OT\$       &$(2, 3)$& 1&C\\
16&PLE\$      &$(3, 1)$& 0&P\\
17&PPLE\$     &$(4, 1)$& 1&A\\
18&PRICOT\$   &$(6, 3)$& 1&A\\
19&RICOT\$    &$(5, 3)$& 0&P\\
20&T\$        &$(1, 3)$& 0&O\\\hline
\end{tabular}
\caption{GSA, LCP, BWT on the reads APPLE, LEMON, APRICOT}
\label{tab:example-defs}
\end{table}

The $i$-th smallest (in lexicographic order) suffix  is denoted $LS[i]$, that is
if $SA[i]=(k,j)$ then $LS[i]= r_{j}[|r_{j}| - k + 1:]$.
In the paper, and especially in the statements, we assume that $R$ is a set of
reads, $\SA$ is the generalized suffix array of $R$, $\LCP$ is the LCP array of
$R$, and $B$ is the Burrows-Wheeler Transform of $R$.

Given a string $Q$ and a collection $R$, notice that all suffixes of $R$
whose prefix is $Q$ appear consecutively in LS.
We call \emph{$Q$-interval}~\cite{bauer_lightweight_2013} on $R$ (or
simply $Q$-interval, if the set $R$ is clear from the context) the
maximal interval $[b, e)$ such
that $Q$ is a prefix of $LS[i]$ for each $i$, $b\le i<e$ (we denote the
$Q$-interval by $q(Q)$).
Sometimes we will need to refer to the $Q$-interval on the set
$R^{r}$: in that case the $Q$-interval is denoted by $q^{r}(Q)$.
We define the \emph{length} and \emph{width} of the $Q$-interval $[b,
e)$ on $R$
as $|Q|$ and the difference $e-b$, respectively.
Notice that the width of the $Q$-interval is equal to the number of
occurrences of $Q$ as a substring of some string $r \in R$.
For instance, on the example in Table~\ref{tab:example-defs} the LE-interval is $[10,12)$.
Whenever the string $Q$ is not specified, we will use the term
\emph{string-interval}.
Since the BWT, the LS, the GSA, and the LCP arrays are all closely
related, a string interval can be viewed as an interval on any of
those arrays.

To extend the previous definition of string interval to consider a
string $Q$ that is a string over the alphabet $\extalphabet$, we have some technical
details to fix, related to the fact that a suffix can contain a sentinel $\$$
only as the last character.
For example, we have to establish what is the $\$A$-interval in
Table~\ref{tab:example-defs}, even though no suffix has $\$A$  as a prefix.
To overcome this hurdle, to each suffix $S[i:]$ of the string $S$ we
associate $S[i:]S[:i-1]$, which is a rotation of $S$, and let $LS'$ be
array of the sorted rotations.
Given a string $Q$ over $\extalphabet$, we define the $Q$-interval as
the maximal interval $[b, e)$ such that $Q$ is a prefix of the $i$-th
rotation (in lexicographic order) for each $i$, $b\le i<e$.

Let $S$ be a string over $\Sigma$.
Then the $S\$$-interval $q(S\$)$ contains exactly one suffix extracted
from each read with suffix $S$.
Moreover, the $\$S$-interval $q(\$S)$ contains exactly one suffix extracted
from each read with prefix $S$.
For this reason, we will say that $q(S\$)$ identifies the set
$R^{s}(S)$ of the reads with suffix $S$, and $q(\$S)$ identifies the set
$R^{p}(S)$ of the reads with prefix $S$.

Let $B^{\rev}$ be the BWT of the set $R^{\rev}=\{ r^{\rev} \mid r\in R\}$, let
$[b,e)$ be the $Q$-interval on $R$ for some string $Q$, and let $[b^{r},e^{r})$ be the
$Q^{\rev}$-interval on $R^{\rev}$.
Then, $[b,e)$ and $[b^{r},e^{r})$ are called \emph{linked} string-intervals.
The linking relation is a 1-to-1 correspondence and two linked intervals
have same width and length, hence $e-b = e^{r}-b^{r}$.
Given two strings $Q_{1}$ and $Q_{2}$, the $Q_{1}$-interval
and the $Q_{2}$-interval on $R$ are either contained one in the other (possibly are the
same) or disjoint.
There are  some interesting relations between string-intervals and the LCP
array.

The interval $[i:j]$ of the LCP array is called an lcp-interval of value $k$
(shortly $k$-interval) if $\LCP[i], \LCP[j+1]< k$, while $\LCP[h]\ge k$ for
each $h$ with $i<h\le j$ and there exists $\LCP[h]=k$ with $i<h\le j$~\cite{Abouelhoda2004}.
An immediate consequence is the following proposition.

\begin{proposition}
\label{proposition:interval-LCP}
Let $R$ be a set of reads, let $\LCP$ be the LCP array of $R$,
let $S$ be a string and let $[b, e)$ be the $S$-interval.
Then $\LCP[h]\ge |S|$  for each $h$ with $b<h\le e-1$.
Moreover, if $S$ is the longest string whose $S$-interval is  $[b, e)$, then
$[b: e - 1]$ is a $|S|$-interval.
\end{proposition}

Proposition~\ref{proposition:interval-LCP} relates the notion of
string-intervals with that of lcp-intervals.
It is immediate to associate to each $k$-interval $[b,e)$ the string $S$
consisting of the common $k$-long prefix of all suffixes in $LS[i]$ with $b\le i< e$.
Such string $S$ is called the \emph{representative} of the $k$-interval.
Moreover, given a  $k$-interval $[b,e)$, we will say that $b$ is its
\emph{opening position} and that $e$ is its \emph{closing position}.

\begin{proposition}
\label{proposition:test-prefix}
Let $S_{1}$, $S_{2}$ be two strings such that the $S_{2}$-interval
$[b_{2}, e_{2})$  is nonempty, and let $[b_{1}, e_{1})$ be the
$S_{1}$-interval.
Then $S_{1}$ is a proper prefix of $S_{2}$ if and only if $[b_{1}, e_{1})$
contains $[b_{2}, e_{2})$ and $|S_{1}|<|S_{2}|$.
\end{proposition}

\begin{proof}
The only if direction is immediate, therefore we only consider the case when $[b_{1},
e_{1})$ contains $[b_{2}, e_{2})$ and $|S_{1}|<|S_{2}|$.

If the containment is proper, the proof is again immediate from the definition
of string-interval.
Therefore assume that  $b_{1}=b_{2}$ and  $e_{1}= e_{2}$.
Since $[b_{1}, e_{1})$ is nonempty, the representative $S$ of $[b_{1}, e_{1})$
is not the empty string.
Moreover both $S_{1}$ and $S_{2}$ are prefixes of $S$, since $S$ is
the longest common prefix of $LS[b,e)$.
Since $|S_{1}|<|S_{2}|$, $S_{1}$ is a proper prefix of $S_{2}$.
\end{proof}

Notice that Proposition~\ref{proposition:test-prefix} is restricted to
nonempty $S_{2}$-intervals, since the $Q$-interval is empty for each
string $Q$ that is not a substring of a read  in $R$.
Therefore relaxing that condition would falsify Proposition~\ref{proposition:test-prefix}.
On the other hand, we do not need to impose the $S_{1}$-interval to be
nonempty, since it is an immediate consequence of the  assumption that $S_{1}$ is a prefix of $S_{2}$.

Given a $Q$-interval and a symbol $\sigma^{\$}\in \Sigma$, the \emph{backward
  $\sigma$-extension} of the $Q$-interval is the $\sigma Q$-interval (that is,
the interval on the GSA of the suffixes sharing the common prefix $\sigma Q$).
We say that a $Q$-interval has a \emph{nonempty} (\emph{empty}, respectively)
backward $\sigma$-extension if the resulting interval has width greater than 0
(equal to 0, respectively).
Conversely, the \emph{forward $\sigma$-extension} of a $Q$-interval is the
$Q \sigma$-interval.
We recall that the FM-index~\cite{Ferragina2005} is essentially made of the
two arrays $C$ and $\Occ$, where
$C(\sigma)$, with $\sigma \in \Sigma$, is the number of occurrences in
$B$ of symbols that are alphabetically smaller than $\sigma$, while
$\Occ(\sigma,i)$ is the number of occurrences of $\sigma$ in the prefix $B[ : i
- 1]$ (hence $\Occ(\cdot, 1)=0$).
It is immediate to obtain from $C$ a function $C^{-1}(i)$ that returns the first
character of $LS[i]$; this fact allows to represent a character with
an integer.
We can now state a fundamental characterization of
extensions of a string-interval~\cite{Ferragina2005,Lam2009}.
This characterization allows to compute efficiently all extensions,
via $C$ and $Occ$.

\begin{proposition}
\label{prop:extension-formulae}
Let $S$ be a string, let $q(S)=[b,e)$ be the $S$-interval and let $\sigma$ be a character.
Then the backward $\sigma$-extension of  $[b,e)$ is $q(\sigma S) = [C(\sigma) +
Occ(\sigma, b) + 1, C(\sigma) + Occ(\sigma, e))$, and the forward
$\sigma$-extension of  $[b,e)$ is $q(S\sigma) = [b + \sum_{c< \sigma} (Occ(\sigma, e) - Occ(\sigma, b)), b + \sum_{c\le \sigma} (Occ(\sigma, e) - Occ(\sigma, b))$.
\end{proposition}

Proposition~\ref{prop:extension-formulae} presents a technical problem
when $\sigma$ is the sentinel $\$$.
More precisely, since all reads share the same sentinel $\$$, we might
not have a correspondence between suffixes in the $q(\$S)$ and reads
with prefix $S$.
More precisely, if  $q(\$S) = [b ,e)$, $b\le i<e$, and $SA[i]=(k,j)$,
we do not know whether $r_{j}$ has prefix $S$ (which is needed to
preserve the correspondence between $q(\$S)$ and $R^{p}(S)$.
For this reason, we sort the suffixes that are equal to the sentinel
$\$$ (corresponding to the positions $i$ such that $SA[i] = (0,
\cdot)$) according to the lexicographic order of the reads.
In other words, we enforce that, for each $i_{1}$, $i_{2}$ where
$SA[i_{1}] = (0, j_{1})$, $SA[i_{2}] = (0, j_{2})$, if $i_{1}<i_{2}$
then $r_{j_{1}}$ lexicographically precedes $r_{j_{2}}$.
For that purpose, it suffices a coordinated scan of the GSA and of the
BWT, exploiting the fact that $B[i] = \$$ and $SA[i] = (k,j)$ iff the
read $r_{j}$ is $k$ long.

\begin{definition}[Overlap graph]
\label{definition:overlap-graph}
Given a set $R$ of reads, its \emph{overlap graph}~\cite{Myers2005} is the
directed graph $G_{O}=(R, A)$ whose vertices are the reads in $R$, and where two reads
$r_i, r_j$ form the arc $(r_i, r_j)$ if they have a nonempty overlap.
Moreover, each arc  $(r_i,r_j)$ of  $G_{O}$
is labeled  by the left extension  $lx_{i,j}$ of $r_i$ with  $r_j$.
\end{definition}

The main use of string graph is to compute the \emph{assembly} of each
path, corresponding to the sequence that can be read by traversing the
reads corresponding to vertices of the path and overlapping those reads.
More formally, given a path
$r_{i_{1}},r_{i_{2}}, \ldots , r_{i_{k}}$ of $G_{O}$, its assembly is
the string
$lx_{i_{1}, i_{2}} lx_{i_{2}, i_{3}} \cdots lx_{i_{k-1}, i_{k}} r_{i_{k}}$.

The original definition of overlap graph~\cite{Myers2005} differs from ours since the
label of the arc $(r_i,r_j)$ consists of the right extension
$rx_{i,j}$ as well as
$lx_{i,j}$.
Accordingly, also their definition of assembly uses the right extensions
instead of the left extensions.
The following lemma establishes the equivalence of those two definitions in terms of \emph{assembly} of a path.

\begin{lemma}
\label{lemma:prefix-interval-path}
Let $G_{O}$ be the overlap graph for $R$ and let
$r_{i_{1}},r_{i_{2}}, \ldots , r_{i_{k}}$ be a path of $G_{O}$.
Then,
$lx_{i_{1}, i_{2}} lx_{i_{2}, i_{3}} \cdots lx_{i_{k-1}, i_{k}} r_{i_{k}} =
r_{i_{k}} rx_{i_{1}, i_{2}} rx_{i_{2}, i_{3}} \cdots rx_{i_{k-1}, i_{k}}$.
\end{lemma}

\begin{proof}
We will prove the lemma by induction on $k$. Let $(r_{h}, r_{j})$ be an arc of
$G_{O}$.
Notice that the path $r_{h} r_{j}$ represents $lx_{h,j} ov_{h,j} rx_{h,j}$.
Since $r_{h} = lx_{h,j} ov_{h,j}$ and
$r_{j}= ov_{h,j} rx_{h,j}$, the case $k=2$ is immediate.

Assume now that the lemma holds for paths of length smaller than $k$ and
consider the path $(r_{i_{1}}, \ldots , r_{i_{k}})$.
The same argument used for $k=2$ shows that
$lx_{i_{1}, i_{2}} lx_{i_{2}, i_{3}} \cdots lx_{i_{k-1}, i_{k}} r_{i_{k}} =
lx_{i_{1}, i_{2}} lx_{i_{2}, i_{3}} \cdots lx_{i_{k-2}, i_{k-1}} r_{i_{k-1}} rx_{k-1,k}$.
By inductive hypothesis
$lx_{i_{1}, i_{2}} lx_{i_{2}, i_{3}} \cdots lx_{i_{k-2}, i_{k-1}} r_{i_{k-1}}
rx_{k-1,k} =
r_{i_{1}} rx_{i_{1}, i_{2}} rx_{i_{2}, i_{3}} \cdots rx_{i_{k-2}, i_{k-1}}
rx_{i_{k-1}, i_{k}}$, completing the proof.
\end{proof}

This definition models the actual use of string graphs to reconstruct a genome~\cite{Myers2005}.
If we have perfect data and no relevant repetitions, the overlap
graph is a directed acyclic graph (DAG) with a unique topological sort, which in
turn reveals a
peculiar structure; the graph is made of tournaments~\cite{diestel_graph_2005}.
More formally, let $<r_{1}, \ldots , r_{n}>$ be the
topological order of $G_{O}$.
If $(r_{i}, r_{j})$ is an arc of $G_{O}$ then also
all $(r_{h}, r_{k})$ with $i<h<k<j$ are arcs of $G_{O}$.
Notice that in this case, all paths from $r_{i}$ to $r_{j}$ have the same assembly.

Less than ideal conditions might violate the previous property.
In fact insufficient coverage (where we do not have reads extracted from some
parts of the original genome) or sequencing errors (where the read is not a
substring of the genome) might result in a disconnected graph, while spurious
overlaps or long
repetitions might result in a graph that is not a DAG.
Nonetheless, the ideal case points out that we can have (and we actually have in
practice) multiple paths with the same assembly.

This suggests that it is possible (and auspicable) to remove some arcs
of the graph without modifying the set of distinct assemblies.
An arc $(r_i, r_j)$ of $G_{O}$  is called \emph{reducible}~\cite{Myers2005}  if
there exists another path from $r_i$ to $r_j$ with the same assembly (\ie the
string $lx_{i,j} r_{j}$).
After removing all reducible arcs we obtain the \emph{string
graph}~\cite{Myers2005}.

In this paper we are going to develop two external-memory algorithms, the first
to  compute the overlap graph associated to a set of reads, and the second to reduce an
overlap graph into a string graph.

For simplicity, and to emphasize that our algorithms are suited also for an
in-memory implementation, we use lists as main data structures.
An actual external-memory implementation will replace such lists with
files that can be accessed only sequentially.
We will use an
array-like notation to denote each element, but accessing those elements sequentially.
Moreover, we will assume that the set of reads $R$ has been processed with the
BCRext algorithm~\cite{bauer_lightweight_2013} to compute the BWT, GSA and LCP of $R$.

\section{Computing the overlap graph}

Our algorithm for computing the overlap graph is composed of two main parts: (i)
computing the unlabeled overlap graph, (ii) labeling the arcs.
Notice that, given a string $S$, the cartesian product $R^{s}(S) \times
R^{p}(S)$, where $R^{p}(S)$, $R^{s}(S)$ are respectively the set of reads in $R$
whose prefix and suffix (respectively) is $S$, consists of  the arcs whose overlap is $S$.
Observe that the pair $(q(S\$), q(\$S))$ of string-intervals represents the set
of arcs whose overlap is $S$, since $q(S\$)$ and $q(\$S))$ represent the sets $R^{s}(S)$ and $R^{p}(S)$, respectively.
Characterizing also the arc labels is more complicated, as pointed out by Definition~\ref{def:arc-interval}.

% The arcs of the overlap graph are labeled by the left extension.
% Let $(r_{i}, r_{j})$ be an arc whose overlap is $S$, and let $P$ be the
% left extension $lx_{i,j}$.

A consequence of Lemma~\ref{lemma:prefix-interval-path} is that we can label each arc with its left extension.
Indeed, given a read $r_i = PS$ we use the
$P$-interval to label  the arcs  $(r_{i}, r_{j})$ (outgoing from $r_i$) with overlap $ov_{i,j}=S$ and
left extension $lx_{i,j}=P$.
Anyway, to compute the $P$-interval we will need also the $PS\$$-interval.
Moreover, our procedure that reduces an overlap graph is
based on a property relating the reverse $P^{r}$ of the left extension $P$;
for this reason we need to encode $P^{r}$ as well as $P$.

\begin{definition}\label{def:arc-interval}
\label{def:arc-label}
Let $P$ and $S$ be two strings.
Then, the tuple $(q(PS\$), q(\$S), q(P),  q^{\rev}(P^{\rev}), |P|, |S|)$ is
the $(P,S)$-\emph{encoding} (or simply encoding) of all arcs with left extension $P$ and overlap $S$.
Moreover, the $(P,S)$-encoding is \emph{terminal} if the $PS\$$-interval has a nonempty
backward $\$$-extension, and it is \emph{basic} if $P$ is the empty string.
\end{definition}

Notice that a basic $(\epsilon, S)$-encoding is equal to $(q(S\$), q(\$S), q(\epsilon),
q(\epsilon), 0, |S|)$, where the interval $q(\epsilon)$ is $[1, n+m+1)$,
where $n + m$ is the overall number of characters in the input reads,
included the sentinels.
Moreover, the differences between basic and non-basic encodings consist of the
information on $P$, that is the arc label.
In other words, the basic encodings already represent the topology of
the overlap graph.
For this reason, the first part of our algorithm will be to compute all basic encodings.
Moreover, we want to read sequentially three lists---\FileBWT, \FileSA and
\FileLCP---that have been previously computed via
BCRext~\cite{bauer_lightweight_2013} containing the BWT $B$, the GSA $\SA$ and the LCP array $\LCP$, respectively.
Another goal of our approach is
to minimize the number of passes over those lists, as a
simpler adaptation of the algorithm of~\cite{Simpson2010}  would require a number
of passes equal to the sum of the lengths of the input reads in the worst case, which would
clearly be inefficient.

% The correctness of our algorithm relies on a property of \FileSA,
% that whenever the same string is the suffix of two different reads,
% the occurrences of the suffix are in the same order as the lexicographic order
% of the reads containing the suffix; we will call a suffix array \emph{consistent}.
% In other words, in a consistent GSA, if the $k$ suffixes of  $j_{1}$
% and of $j_{2}$ are the same string, then $i_{1}<i_{2}$ iff $r_{j_{1}}$
% lexicographically precedes $r_{j_{2}}$.
% Notice that a GSA is not necessarily consistent, but we will show how to
% modify a suffix array to guarantee consistency with two linear scans of the GSA.

% We will show that two scans of the GSA suffice to obtain such a GSA.

% The first scan extracts the sequence of pairs $(1,j)$, \ie the entries of the
% GSA corresponding to the reads, hence obtaining the reads of $R$ sorted
% lexicographically.
% The second scan uses the sorted $R$ to
% reorder consecutive entries of the GSA sharing the same suffix.
% This ordering will be essential in the following since a particular
% operation (namely, the backward $\$$-extension, as defined below) is
% possible only if this specific  order is assumed. \notaestesa{GDV}{ Richiede di
%   tenere in memoria un array di $|R|$ elementi}

\begin{figure}[htb]
\centering
\begin{tikzpicture}[thick]
\node (r1) at (-1,2) {$r_{1}$};
\node (r2) at (-1,1) {$r_{2}$};
\node (r1p) at (0.9,2) {ATATCATC};
\node (r1s) at (3.8,2) {GATCTACTATTAC};
\node (r2s) at (3.8,1) {GATCTACTATTAC};
\node (r2x) at (6.7,1) {TTCATATC};
\node (s1pt) at (0.9,0.1) {$P=lx_{1,2}$};
\node (s2st) at (3.8,0.1) {$S=ov_{1,2}$};
\node (s2xt) at (6.7,0.1) {$X=rx_{1,2}$};
\draw (-0.3, 2.3) rectangle (5.6,1.7);
\draw (2, 1.3) rectangle (7.8,0.7);
\draw (2,2.2) edge [dashed] (2,-0.3);
\draw (5.6,2.2) edge [dashed] (5.6,-0.3);
\draw (-0.3,2.2) edge [dashed] (-0.3,-0.3);
\draw (7.8,1.2) edge [dashed] (7.8,-0.3);
\end{tikzpicture}
\caption{Example of arc of the overlap graph.
The read $r_{1}$ is equal to ATATCATCGATCTACTATTAC, while the read $r_{2}$ is
equal to GATCTACTATTACTTCATATC.
}
\label{fig:arc-overlap}
\end{figure}
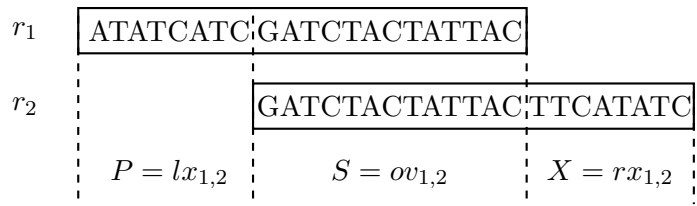

\begin{definition}
\label{definition:seed}
Let $S$ be a proper substring of some read of $R$.
If both the $S\$$-interval and the $\$S$-interval are nonempty, then
$S$ is called a \emph{seed}.
\end{definition}

Finding all basic encodings is mostly equivalent to finding all seeds of $R$.
Now we will prove that  a position $p$ can be the opening position of several
lcp-intervals, but it
can be the opening position of only one seed.

\begin{lemma}
\label{lemma:1-seed-opening-position}
Let $S$ be a $k$-long seed and let $[b,e)$ be the corresponding $k$-interval.
Then $k= \FileLCP[b+1]$ and $k > \FileLCP[b]$.
\end{lemma}

\begin{proof}
Since $S$ is a seed, $S$ has a forward $\$$-extension.
Moreover $\$$ is the smallest character, hence $S\$$ is the first suffix of the
$S$-interval, that is the $LS[b]=S\$$.
Since $S$ is a seed,  $S$ is not an entire read.
Hence the $\$S$-interval and the $S\$$-interval are both nonempty and are disjoint.
Consequently, $S$ is a prefix of $LS[b+1]$, that is $LS[b+1]=S\alpha$
for some string $\alpha$ on the alphabet $\Sigma\cup\{\$\}$.
By definition, $\FileLCP[b+1]$ is the length of the longest common prefix of
$S\$$ and $S\alpha$, that is $\FileLCP[b+1]=|S|$.

If $b=1$, by definition $\FileLCP[b]=-1$, hence $\FileLCP[b] < |S|$.
If $b>1$, $S$ is not a prefix of $LS[b]$, hence $\FileLCP[b] < |S|$.
\end{proof}

\begin{corollary}
\label{corollary:seed-has-one-opening}
Let $[b,e)$ be an interval.
Then $b$ is the opening position of at most one seed.
\end{corollary}

In the presentation of our algorithm,
we  need a simple procedure called \textbf{Merge}.
This procedure operates on
lists of encodings $([b,e), \cdot, \cdot, \cdot,  \cdot,  \cdot)$ that
are sorted by increasing $b$ (records with the same $b$ are sorted by
decreasing $e$).
We do not actually need to write a new list, as \textbf{Merge}
consists of choosing the list from which to read the next record.

To compute all basic encodings,
the procedure \textbf{BuildBasicArcIntervals} (Algorithm~\ref{alg:BAI}) reads
sequentially the lists \FileBWT,
\FileLCP and \FileSA while keeping in $\#_\$$ the number of sentinels in the portion of BWT before
the current position.
Moreover we maintain a stack $Z$ that is used to store the
relevant $k$-intervals whose opening position has been read, but whose
closing position has not.

\begin{algorithm}[htb!]
\SetKwInOut{PRE}{Input}
\SetKwInOut{POST}{Output}
\PRE{Three lists \FileBWT, \FileLCP, and \FileSA containing the BWT, the LCP array and the GSA of the set $R$, respectively.}
\POST{%
  A set of lists $\FileE(\sigma, \l_{S}, \l_{S})$ each containing the
  $(\epsilon, S)$-encoding for the seeds $S$ whose first character is $\sigma$
  and whose length is $\l_{S}$.
The encodings $([b, e), \cdot, \cdot, \cdot, \cdot, \cdot)$ in each list are
sorted by increasing values of $b$.
}

$\#_{\$} \gets 0$\;
\If{$\FileBWT[1] = \$$}{
    $\#_{\$} \gets 1$\label{alg:BAI:sharp-init}\;
  }
  $p\gets 2$\;
$Z \gets$ empty stack\;
\While{$p \le |\FileLCP|$}{
  \If{ $\FileLCP[p] < \FileLCP[p-1]$\label{alg:BAI:closing}}{\label{alg:BAI:begin-close}
    $([b, e_1), \l_S, b_{\$}) \gets top(Z)$\;
    \While{$Z$ is not empty and $\l_S > \FileLCP[p]$}{
      \If{$\#_{\$} > $ $b_{\$}$}{
        append $([b, e_1), [b_{\$}$, $\#_{\$})$, $[1, |\FileBWT| +1)$, $[1,
        |\FileBWT| +1)$, $0$, $\l_S)$ to the list $\FileE(C^{-1}(p), \l_S, \l_S)$\label{alg:BAI:output}\;
      }
      pop(Z)\;
      $([b, e_1), \l_S, b_{\$}) \gets top(Z)$\label{alg:BAI:end-close}\;
    }
  }

  \If{$\FileLCP[p] > \FileLCP[p-1]$\label{alg:BAI:test-opening}}{\label{alg:BAI:begin-open}
    $(k,j)\gets \FileSA[p-1]$\label{alg:BAI:fix-k}\;
    $(k^{*},j^{*})\gets \FileSA[p]$\;
    $q \gets p$\;
\If{$\FileLCP[p] = k$\label{alg:BAI:test-S-dollar}}{%
    \While{$\FileLCP[q + 1] = \FileLCP[p] = k^{*} = k$\label{alg:BAI:begin-Sdollar}}{
        $q \gets q+1$\;
        $(k^{*},j^{*})\gets \FileSA[q]$\;
      }
      push $( [p-1, q+1), \FileLCP[p], \#_{\$})$ to $Z$\label{alg:BAI-push}\;
      $p \gets q$\;
      \label{alg:BAI:end-open}
    }
}
  \If{$\FileBWT[p] = \$$} {
    $\#_{\$} \gets \#_{\$}+1$\label{alg:BAI:sharp-inc}\;
  }
  $p \gets p+1$\;
}
\While{$Z$ is not empty\label{alg:BAI:begin-empty-stack}}{
    $([b, e_1), \l_S, b_{\$}) \gets top(Z)$\;
      \If{$\#_{\$} > $ $b_{\$}$}{
        append $([b, e_1), [b_{\$}$, $\#_{\$})$, $[1, |\FileBWT| +1)$, $[1,
        |\FileBWT| +1)$, $0$, $\l_S)$ to the list $\FileE(C^{-1}(p), \l_S, \l_S)$\;
      }
      pop(Z)\;\label{alg:BAI:end-empty-stack}
}
\caption{BuildBasicArcIntervals}
\label{alg:BAI}
\end{algorithm}

When the current  position is $p$, the only interesting cases are if
$p-1$ is an opening position or $p$ is a closing position.
In the first case (see lines~\ref{alg:BAI:begin-open}--\ref{alg:BAI:end-open}),
Lemma~\ref{lemma:1-seed-opening-position} and
Corollary~\ref{corollary:seed-has-one-opening} show that only the
$\FileLCP[p]$-long interval whose opening position is $p-1$ might have a
seed as representative.
Let $S$ be the representative of such   $\FileLCP[p]$-long interval.
First, at lines~\ref{alg:BAI:begin-Sdollar}--\ref{alg:BAI:end-open},
we compute the $S\$$-interval $[p-1, e_1)$ (since such string interval is an
initial portion of the interval, we only need to read some records from the
input lists).
If the $S\$$-interval is nonempty, then it is pushed onto $Z$ together with the
current value of $\#_\$$ and the length of $S$.
%We can easily compute the function $\sigma(i)$ that reads the function
%$C$ of the FM-index to determine the first character of the suffix
%represented by $\SA[i]$.
Notice that the closing position of the
seed is currently unknown and will be determined only later, but the
information in $Z$ will suffice (together with some information
available only when closing the interval) to reconstruct the basic
encoding relative to the seed $S$, \ie the $(\epsilon, S)$-encoding.

In the second case, $p$ is a closing position
(lines~\ref{alg:BAI:begin-close}--\ref{alg:BAI:end-close}) and the procedure removes
from the stack $Z$ all the records $([b, e_1), \l_S, b_\$)$
corresponding to an $S$-interval $[b,p)$ whose opening position is $b$ and whose
forward $\$$-extension is $[b, e_1)$.
The backward $\$$-extension is $[b_{\$}, \#_\$)$, since $b_{\$}$ is the number
of sentinels in $B[:b-1]$, while $\#_\$$ is
the number of sentinels in $B[:p-1]$.
Clearly $[b_{\$}, \#_\$)$ is nonempty (and $S$ is  a seed) if and
only if $\#_\$ > b_\$$.

Notice that the stack $Z$ always contains a nested hierarchy of distinct
seeds (whose ending position might be currently unknown),
that all $k$-intervals whose closing position is $p$ are exactly the
intervals with $k>\FileLCP[p]$, and they are found at the top of $Z$.

After all iterations, the stack $Z$ contains the intervals
whose closing position is $p=|\FileLCP|$.
Those intervals are managed at
lines~\ref{alg:BAI:begin-empty-stack}--\ref{alg:BAI:end-empty-stack}.

There is a final technical detail:
each output basic encoding associated to the overlap $S$ is output to the list
$\FileE(S[1], |S|, |S|)$.
In fact we will use some different lists $\FileE(\sigma, \l_{S}, \l_{PS})$,
each containing the encodings corresponding to seed $S$ and
left extension $P$, where $\sigma$ is the first character of $PS$,
$\l_{S}=|S|$ and $\l_{PS}=|P|+|S|$.

Moreover, a list $\FileE(\sigma, l_{s}, l_{PS})$ is \emph{correct} if
it contains exactly the $(P,S)$-encodings such that $|S|=\l_{S}$ and
$|S|+|P|=\l_{PS}$ and the encodings $([b,e), \cdot, \cdot, \cdot,
\cdot, \cdot)$ are sorted by increasing values of $b$.

\begin{lemma}
\label{lemma:all-seeds-are-found}
Let $R$ be a set of reads, and let $S$ be a seed of $R$.
Then the $(\epsilon, S)$-encoding
$(q(S\$), q(\$S), [1, |\FileBWT| +1), [1,|\FileBWT| +1), 0, |S|)$ is output by Algorithm~\ref{alg:BAI}.
\end{lemma}

\begin{proof}
Let $[b_{S}, e_{S})$ be the $S$-interval, let  $[b_{S\$}, e_{S\$})$ be the $S\$$-interval, and
let  $[b_{\$S}, e_{\$S})$ be the $\$S$-interval.
Since $S$ is a seed, all those intervals are nonemtpy.
Moreover, the sentinel is the smallest character, hence $b_{S\$}=b_{S}$.

When $p=b_{S}+1$, since $S$ is not an entire read (by definition of seed), $S$
is a prefix of both $LS[b_{S}+1]$ and  $LS[b_{S}]$, hence
$\FileLCP[p]\ge |S|$.
Moreover, since the $S\$$-interval is not empty,  $S\$$ is a prefix of
$LS[b_{S}]$ hence $\FileLCP[p]= |S|$, as the sentinel is not part of a common
prefix.

By definition of $S$-interval, $\FileLCP[p-1]< |S|$, hence the condition at
line~\ref{alg:BAI:test-opening} is satisfied and at line~\ref{alg:BAI:fix-k}
$k = \FileLCP[p] = |S|$.
When reaching line~\ref{alg:BAI-push}, $\FileLCP[x] = k =  LS[p-1]$ for all $x$ with $p\le x\le q$.
All those facts and the observation that $S\$$ is a prefix of
$LS[p-1]$, imply that  $S\$$ is a prefix of all suffixes $LS[i]$ with $p-1 \le
i\le q$.
Since the while condition does not currently hold (as we have exited from the
while loop),  $S\$$ is not a prefix of $LS[q+1]$, hence $[p-1, q+1)$ is
the $S\$$-interval.
Consequently at line~\ref{alg:BAI:test-opening} we push the triple $(q(S\$),
|S|, \#_{\$})$ on $Z$, where $\#_{\$}$ is the number of sentinels in
$\FileBWT[:p]$.

We distinguish two cases: either $e_{S}\le n$ or  $e_{S}> n$.
If $e_{S}\le n$ then there is an iteration where $p=e_{S}$.
During such iteration the condition at line~\ref{alg:BAI:closing} holds, hence all
entries $(q(T\$), |T|, T_{\$})$ at the top of $Z$ such that $|T| > \FileLCP[p]$
are popped and the corresponding encoding is output at line~\ref{alg:BAI:output}.
Since $[b_{S}, e_{S})$ is an $S$-interval,  $|S| > \FileLCP[p]$ hence the
interval $(q(S\$), |S|, \#_{\$})$ is popped.
Since $\$$ is the first symbol of the alphabet,
$q(\$S)=[b_{S\$}, e_{S\$})=[b_{\$}, \#_{\$})$.

If $e_{S\$} > |\FileLCP|$, then the condition of line~\ref{alg:BAI:closing} is
never satisfied.
Anyway, the stack $Z$ is completely emptied at
lines~\ref{alg:BAI:begin-empty-stack}--\ref{alg:BAI:end-empty-stack} and the
same reasoning applies to show that the $(\epsilon, S)$-encoding
is output.
\end{proof}

\begin{lemma}
\label{lemma:only-seeds-are-found}
Let $R$ be a set of reads, and let
$([b, e_1), [b_{\$}$, $\#_{\$}), [1, |\FileBWT| +1), [1,|\FileBWT| +1), 0, \l_{S})$ be an encoding
output by Algorithm~\ref{alg:BAI}.
Then $q(S\$) = [b, e_1)$, $q(\$S) = [b_{\$}$, $\#_{\$})$, $\l_{S}=|S|$ for some
seed $S$ of $R$.
\end{lemma}

\begin{proof}
Notice that encodings are output only if a triple  $( [p-1, q+1), \FileLCP[p],
\#_{\$})$ is pushed on $Z$, which can happen only if $\FileLCP[p] >
\FileLCP[p-1]$.
By Lemma~\ref{lemma:1-seed-opening-position}, $p-1$ can be the opening
position only of the seed $S$ obtained by taking the $\FileLCP[p]$-long prefix
of  $LS[p]$.
By the condition at line~\ref{alg:BAI:test-S-dollar}, the triple
$( [p-1, q+1), \FileLCP[p], \#_{\$})$ is pushed on $Z$ only if the
$S\$$-interval is not empty.

Since $\#_{\$} > b_{\$}$ and  by the value of $\#_{\$}$, also the $\$S$-interval
is nonempty, hence $S$ is a seed.
\end{proof}

\begin{lemma}
\label{lemma:1-list-disjoint-intervals}
Let $\FileE(\sigma, \l_{PS}, \l_{PS})$ be a list output by Algorithm~\ref{alg:BAI}.
Let $f_{1}=(q(S_{1}\$), q(\$S_{1}), \cdot,  \cdot, \cdot, |S_{1}|)$ and
$f_{2}=(q(S_{2}\$), q(\$S_{2}),  \cdot,  \cdot, \cdot, |S_{2}|)$ be two
encodings in $\FileE(\sigma, 0, \l_{PS})$,
$q(S_{1}\$) = [b_{1}, e_{1})$ and $q(S_{2}\$) = [b_{2}, e_{2})$.
Then the intervals $q(S_{1}\$)$ and $q(S_{2}\$)$ are disjoint.
Moreover, $f_{1}$ precedes $f_{2}$ in $\FileE(\sigma, \l_{S}, \l_{PS})$ iff
$b_{1}< b_{2}$.
\end{lemma}

\begin{proof}
By construction, $\l_{S}=|S_{1}|=|S_{2}| = \l_{PS}$
and $\sigma = S_{1}[1]=S_{2}[1]$.
Since $|S_{1}|=|S_{2}|$, the two string-intervals $q(S_{1}\$)$
and $q(S_{2}\$)$ cannot be nested, hence they are disjoint.

Notice that, since  $[b_{1}, e_{1})$
and $[b_{2}, e_{2})$ are disjoint, then $b_{1}< e_{1} \le b_{2}$ or
$b_{1}\ge e_{2}> b_{2}$.
Assume that $b_{1}< e_{1} \le b_{2}$ and let us
consider the iteration when $p=e_{1}$, \ie when $f_{1}$ is output.
Since $e_{1} \ge b_{2}$, the entry $([b_{2}, e_{2} - 1], \cdot, \cdot)$ has not
been pushed to $Z$ yet, hence
$f_{1}$ precedes $f_{2}$ in $\FileE(\sigma, 0, \l_{PS})$.

If $b_{1}\ge e_{2}> b_{2}$ the same argument shows that $f_{2}$ precedes $f_{1}$
in $\FileE(\sigma, 0, \l_{PS})$, completing the proof.
\end{proof}

\begin{corollary}
\label{corollary:basic-encodings-correct}
Let $\FileE(\sigma, 0, \l_{PS})$ be a list computed by Algorithm~\ref{alg:BAI}.
Then $\FileE(\sigma, 0, \l_{PS})$ contains exactly the $(\epsilon,
S)$-encodings of all seed $S$ such that $\sigma=S[1]$.
\end{corollary}

There is an important observation on the sorted lists of encoding that we will manage.
Let $f_{1}=(q(P_{1}S_{1}\$), q(\$S_{1}), q(P_{1}),  q^{\rev}(P_{1}^{\rev}),
|P_{1}|, |S_{1}|)$ and $f_{2}=(q(P_{2}S_{2}\$), q(\$S_{2}), q(P_{2}),  q^{\rev}(P_{2}^{\rev}),
|P_{2}|, |S_{2}|)$ be two encodings that are stored in the same list
$\FileE(\sigma, \l_{S}, \l_{PS})$, hence $|S_{1}|=|S_{2}|$ and  $|P_{1}S_{1}|=|P_{2}S_{2}|$.
Since $|P_{1}S_{1}|=|P_{2}S_{2}|$, the two string-intervals $q(P_{1}S_{1}\$)$
and $q(P_{2}S_{2}\$)$ are disjoint (as long as we can guarantee that
$P_{1}S_{1}\neq P_{2}S_{2}$), hence sorting them by their opening position
of the interval implies sorting  also by closing position.

\subsection{Labeling the overlap graph.}

To complete the encoding of each arc, we need to compute the left extension.
Such step will be achieved with the  \textbf{ExtendEncodings} procedure
(Algorithm~\ref{alg:EAI}), where the $(P,S)$-encodings are elaborated,
mainly via backward $\sigma$-extensions, to
obtain the $(\sigma P,S)$-encodings.
Moreover, when $PS$ has a nonempty backward $\$$-extension, we have determined
that the encoding is terminal, hence we output the arc encoding to the lists
$\FileArc(|P|, z)$ which will contain the arc encodings of the arcs
incoming in the read $r_{z}$ and whose
left extension has length $|P|$.

The first fundamental observation is that a  $(P, S)$-encoding
can be obtained by extending a $(\epsilon, S)$-encoding with (if $|P| = 1$), or by extending
a $(P_{1}, S)$-encoding (if $|P| > 1$ and $P_{1} = P[2:]$).
Those extensions are computed in two phases: the first phase computes
all partially extended encodings $(q(\sigma PS\$), q(\$S), q(P),  q^{\rev}(P^{\rev}),
|P|, |S|)$ of the $(P,S)$-encodings.
The second phase starts from the partially extended encodings and
completes the extensions obtaining all $(\sigma P, S)$-encodings.

We iterate the
procedure \textbf{ExtendEncodings} for increasing
values of $\l_{PS}$, where each step scans all lists
$\FileE(\cdot, \cdot, \l_{PS})$, and writes
the lists $\FileP(\cdot, \cdot, \l_{PS} + 1)$, by
computing all backward  $\sigma$-extensions.
The lists are called \FileP{} as a
mnemonic for the fact that those encodings have been extended only partially.
Those lists will be then fed to the \textbf{CompleteExtensions} procedure to
complete the extensions, storing the result in the lists $\FileE(\cdot, \cdot, \l_{PS} + 1)$.

\begin{algorithm}[tb!]
\SetKwInOut{PRE}{Input}
\SetKwInOut{POST}{Output}
\PRE{
  Two lists \FileBWT and \FileSA containing the BWT and the GSA of the set $R$,   respectively.
  The correct lists $\FileE(\cdot, \cdot, \cdot)$ containing the $(P,S)$-encodings
  such that $|P| = \l_{P}$.
}
\POST{
  The correct lists $\FileP(\cdot, \cdot, \cdot)$ containing the
  partially extended $(\sigma P, S)$-encoding.
  The arcs of the overlap graph outgoing from reads of length
  $\l_{PS}-1$, incoming in a read $r_{z}$, and
  with
  left extension long $\l_{P}$ are   appended to the file
  $\FileArc(\l_{P}, z)$.
}
$\Pi, \pi \gets$ $|\Sigma|$-long vectors $\bar{0}$\;
$p\gets 0$\;
\ForEach{%
  $([b,e), q(\$ S), q(P), q^{r}(P^{r}), l_P, \l_{PS}) \in$
Merge$(\{ \FileE(\sigma, \l_{PS} - \l_{P}, \l_{PS}) :
\sigma\in\Sigma, \l_{PS} \ge \l_{P}\})$}{%
  $\Pi(\sigma) \gets \Pi[\sigma]+\pi[\sigma]$, for each $\sigma \in \Sigma$\label{alg:EAI-Pi-value}\;
  \While{$p<b$\label{alg:EAI-update-Pi-begin}}{%
    $\Pi[B[p]] \gets \Pi[B[p]] +1$\;
    $p\gets p+1$\label{alg:EAI-update-Pi-end}
  }
  $\pi \gets\bar{0}$\label{alg:EAI-reset-pi}\label{alg:EAI-Pi-value-fixed}\;
  \For{$p\gets b$ \KwTo $e-1$\label{alg:EAI-for-pi}}{%
    \If{$\FileBWT[p] \ne \$$}{%
      Increment $\pi[\FileBWT[p]]$ by $1$\label{alg:EAI-increment-pi}\;
    }
    \If{$\FileBWT[p] = \$$ and $\l_{P}>0$\label{alg:EAI-nonemtpy-P}}{%
      $(k, j) \gets \SA[p]$\;
      \ForEach{read $r_{z}\in q(\$ S)$}{%
        Append the arc $\langle j, i, q^{r}(P^{r})\rangle$ to
        $\FileArc(\l_{P}, z)$\label{alg:EAI-nonemtpy-P-end}\;
      }
    }
  }
  \ForEach{$\sigma \in \Sigma$ such that $\pi[\sigma] > 0$\label{alg:EAI-extension}}{%
    \If{$\l_{P} = \l_{PS}$}{%
      $q(P)\gets q(\sigma)$; $q^{r}(P^{r}) \gets q(\sigma)$
    }
    $b' \gets C[\sigma] + \Pi[\sigma] + 1$\label{alg:EAI-b}\;
    $e' \gets b' + \pi[\sigma]$\label{alg:EAI-e}\;
    Append $([b',e'), q(S\$), q(P), q^{r}(P^{r}), l_P+1, \l_{PS}+1)$ to
    $\FileP(C^{-1}(b'), \l_{PS}-\l_{P}, \l_{PS}+1)$\label{alg:EAI-extension-end}\;
  }
}
\label{alg:EAI-extend-P}
\caption{ExtendEncodings$(\l_{P})$}
\label{alg:EAI}
\end{algorithm}

The procedure \textbf{ExtendEncodings} basically extends a sequence
of $PS\$$-intervals $[b,e)$ that are sorted by increasing values of $b$.
The procedure consists of a few parts; up to
line~\ref{alg:EAI-increment-pi} the procedure maintains the arrays $\Pi$ and
$\pi$ that are respectively equal to the number of occurrences of each symbol
$\sigma$ in $\FileBWT[:b-1]$ (resp. in $\FileBWT[b:p-1]$).
The correctness of this part is established by Lemmas~\ref{lemma:Pi-meaning},~\ref{lemma:pi-meaning}.

Lines~\ref{alg:EAI-nonemtpy-P}--\ref{alg:EAI-nonemtpy-P-end} determine if the
representative of $[b,e)$ corresponds to an entire read, \ie if the current
encoding is terminal; in that case some arcs of
the overlap graph have been found and are output to the appropriate list.

The third part (lines~\ref{alg:EAI-extension}--\ref{alg:EAI-extension-end})
computes all backward $\sigma$-extensions of the current
$PS\$$-interval $[b,e)$, obtaining the partially extended encodings.
At line~\ref{alg:EAI-extend-P} we call the procedure that completes the
extensions of the encodings.

In the following we will say that a list $\FileE(\sigma, \l_{S},
\l_{PS})$ of encodings is correct if it contains exactly all $(P,S)$-encodings such
that $\sigma$ is the first character of $PS$, $\l_{P} = |P|$, and $\l_{PS} = |PS|$.
Moreover the encodings  $([b,e), \cdot, \cdot, \cdot,  \cdot,  \cdot)$
are sorted by increasing values of $b$.

A list  $\FileP(\sigma, \l_{S},
\l_{PS})$ of partially extended encodings is correct if it contains
exactly all partially extended $(\sigma P,S)$-encodings such
that $\l_{P} = |P|$, and $\l_{PS} = |PS|$.
Moreover the partially extended encodings  $([b,e), \cdot, \cdot,
\cdot,  \cdot,  \cdot)$
are sorted by increasing values of $b$.

Finally, we would like to point out that each list $\FileE(\cdot, \cdot,
\cdot)$ and $\FileP(\cdot, \cdot, \cdot)$ contains disjoint intervals.
If we can guarantee that the intervals in each list are sorted in
non-decreasing order of the end boundary (we will prove this property of
\textbf{CompleteExtensions}), then those intervals are also  sorted in
non-decreasing order of the start boundary (as required for the correctness of
successive iterations of \textbf{ExtendEncodings}).

\begin{lemma}
\label{lemma:pi-meaning}
Let the lists $\FileE(\cdot, \cdot, \l_{PS})$ be the input of
Algorithm~\ref{alg:EAI} and assume that all those lists are correct.
Let $([b,e) = q(PS\$), q(\$ S), q(P), q^{r}(P^{r}), l_P, \l_{PS})$
be the current encoding.
Then
at line~\ref{alg:EAI-extension} of Algorithm~\ref{alg:EAI}, $\pi[\sigma]$ is
equal to the number of occurrences of $\sigma$ in $\FileBWT[b:e-1]$.
\end{lemma}

\begin{proof}
Notice that $\pi$ is reset to zero at line~\ref{alg:EAI-reset-pi} and is
incremented only at line~\ref{alg:EAI-increment-pi}.
The condition of the for loop (line~\ref{alg:EAI-for-pi}) implies the lemma.
\end{proof}

\begin{lemma}
\label{lemma:Pi-meaning}
Let the lists $\FileE(\cdot, \cdot, \l_{PS})$ be the input of
Algorithm~\ref{alg:EAI} and assume that all those lists are correct.
Let $([b,e) = q(PS\$), q(\$ S), q(P), q^{r}(P^{r}), l_P, \l_{PS})$
be the current encoding.
Then
at line~\ref{alg:EAI-Pi-value-fixed} of Algorithm~\ref{alg:EAI}, $\Pi[\sigma]$ is
equal to the number of occurrences of $\sigma$ in $\FileBWT[:b-1]$.
\end{lemma}

\begin{proof}
We will prove the lemma by induction on the number of the encodings
that have been read.
When extending the first input encoding, $\pi$ consists of zeroes and
the lemma holds, since the while loop at
lines~\ref{alg:EAI-update-Pi-begin}--\ref{alg:EAI-update-Pi-end} increments
$\Pi$ by the number of occurrences of each symbol $\sigma$ up to $b-1$.

Let $k$ be the number of encodings that have been read, with $k\ge 2$,
and let $([b_{1},e_{1}), \cdot, \cdot, \cdot, \cdot, \cdot)$  be the
$(k-1)$-th encoding read.
By Lemma~\ref{lemma:pi-meaning}, $\pi[\sigma]$ is
equal to the number of occurrences of $\sigma$ in $\FileBWT[b_{1},
e_{1})$, hence after line~\ref{alg:EAI-Pi-value} $\Pi$
contains the number of occurrences of each symbol in $\FileBWT[:p-1]$.
The while loop at
lines~\ref{alg:EAI-update-Pi-begin}--\ref{alg:EAI-update-Pi-end} increments
$\Pi$ by the number of occurrences of each symbol $\sigma$ in the portion of
$\FileBWT$ between $e_{1}$ and $b-1$, completing the proof.
\end{proof}

\begin{lemma}
\label{lemma:ExtendEncodings-correct}
Let the lists $\FileE(\cdot, \cdot, \l_{PS})$ be the input of
Algorithm~\ref{alg:EAI} and assume that all those lists are correct.
Let $([b,e) = q(PS\$), q(\$ S), q(P), q^{r}(P^{r}), l_P, \l_{PS})$
a generic input encoding, and let $\sigma$ be a character of $\Sigma$.
If $[b,e)$ has a nonempty
backward $\sigma$-extension $[b_{1},e_{1})$, then \textbf{ExtendEncodings}
outputs the partially extended encoding
$([b_{1},e_{1}), q(\$S), q(P), q^{r}(P^{r}), \l_P, \l_{PS}+1)$ to the list
$\FileP(\sigma, |S|, \l_{PS + 1})$.
\end{lemma}

\begin{proof}
By  Lemma~3.1 in~\cite{Ferragina2005}, the backward $\$$-extension of $[b,e)$
is equal to $[C[\sigma] + Occ(\sigma, b) + 1, C[\sigma] + Occ(\sigma, e) )$.
By Lemmas~\ref{lemma:Pi-meaning},~\ref{lemma:pi-meaning}, the values of $b'$ and
$e^{r}$ computed at lines~\ref{alg:EAI-b}--\ref{alg:EAI-e} is correct.

Notice that $\pi[\sigma]>0$ iff $Occ(\sigma, e) > Occ(\sigma, b) + 1$, hence the
partially extended encoding $([b_{1},e_{1}), q(\$S), q(P), q^{r}(P^{r}), l_P, \l_{PS}+1)$
is output iff $e_{1} > b_{1}$, that is the backward $\sigma$-extension is nonempty.
\end{proof}

\begin{lemma}
\label{lemma:ExtendEncodings-correct-all}
Let $([b_{1},e_{1}), q(\$S), q(P), q^{r}(P^{r}), \l_P, \l_{PS}+1)$ be a partially extended
encoding that is output by \textbf{ExtendEncodings}.
Let $\sigma = C^{-1}(b_{1})$.
Then $([b,e), q(\$S), q(P), q^{r}(P^{r}), \l_P, \l_{PS})$ is an encoding in
$\FileP(\sigma, |S|, \l_{PS})$ and $[b_{1},e_{1})$ is the backward
$\sigma$-extension of $[b,e)$.
\end{lemma}

\begin{proof}
Let $([b,e), q(\$S), q(P), q^{r}(P^{r}), \l_P, \l_{PS})$ be the
current input encoding when $([b_{1},e_{1}), q(\$S), q(P),
q^{r}(P^{r}), \l_P, \l_{PS}+1)$ is output by
\textbf{ExtendEncodings}.
When computing the partially extended interval
(line~\ref{alg:EAI-extension-end}), $b_{1} = C[\sigma] + \Pi[\sigma] +
1$ and $e_{1} = b_{1} + \pi[\sigma]$.
The lemma is a direct consequence of
Proposition~\ref{prop:extension-formulae} and Lemmas~\ref{lemma:pi-meaning},~\ref{lemma:Pi-meaning}.
\end{proof}

\begin{lemma}
\label{lemma:sorted-extend-arc}
Let $\FileP(\sigma, \l_{PS}-\l_{P}, \l_{PS}+1)$ be any list written by ExtendEncodings.
Then  the encodings
$([b',e'), q(S\$), q(P), q^{r}(P^{r}), \l_P, \l_{PS})$ in
$\FileE(\sigma, \l_{PS}-\l_{P}, \l_{PS}+1)$ are sorted by increasing values of
$e'$.
Moreover the intervals $[b',e')$ are disjoint.
\end{lemma}

\begin{proof}
By
Lemmas~\ref{lemma:ExtendEncodings-correct}~\ref{lemma:sorted-extend-arc},
the list $\FileP(\sigma, \l_{PS}-\l_{P}, \l_{PS}+1)$ contains only partially extended encodings
relative to the pairs $(P,S)$ where $\sigma$ is the first symbol of $PS$.

Let us now consider a generic partially extended encoding
$([b',e'), q(S\$), q(P), q^{r}(P^{r}), \l_{P}, \l_{PS})$
written to $\FileP(\sigma, \l_{PS}-\l_{P}, \l_{PS}+1)$.
Notice that $e'-b'=\pi[\sigma]$ and that, while managing the next partially extended
encoding, $\Pi[\sigma]$ will be incremented by $e'-b'$, hence  during the
next iterations of the foreach loop $C[\sigma] +
\Pi[\sigma] + 1$ will be at least as
large as the value of $e'$ at the current iteration.
That completes the proof, since the start boundary is always equal to $C[\sigma]
+ \Pi[\sigma] + 1$ (see line~\ref{alg:EAI-b}).
\end{proof}

The following corollary summarizes this subsection.

\begin{corollary}
\label{corollary:partial-extensions-correct}
Let the lists $\FileE(\cdot, \cdot, \l_{PS})$ be the input of
Algorithm~\ref{alg:EAI} and assume that all those lists are correct.
Let $([b,e) = q(PS\$), q(\$ S), q(P), q^{r}(P^{r}), l_P, \l_{PS})$
be the current encoding.

Then \textbf{ExtendEncodings} produces
the correct lists $\FileP(\sigma, \l_{s}, \l_{PS}+1)$.
\end{corollary}

\begin{lemma}
\label{lemma:terminal-arc}
If the input encodings are correct, then
Algorithm~\ref{alg:EAI} outputs the arc $(r_{j}, r_{i})$ to
$\FileArc(\l_{P}, i)$ iff there exists a read $r_{j}=PS$ with $P$ and $S$
both nonempty and there exists a read $r_{i}$ whose prefix is $S$.
\end{lemma}

\begin{proof}
Notice that the encoding of the arc  $(r_{j}, r_{i})$ is output to
$\FileArc(\l_{PS} + 1, i)$ only if we are
currently extending the $(P,S)$-encoding (hence $S$ is a seed) and we have found that the
$PS\$$-interval has  a nonempty backward $\$$-extension, since $B[p]$ is the symbol
preceding $PS$ in a suffix and $B[p]=\$$.
By definition of seed, $S$ is nonempty and there exists a read $r_{i}$
with prefix $S$, while $P$ is
nonempty by the condition at line~\ref{alg:EAI-nonemtpy-P}.

Assume now that $r_{j}=PS$ is a read with $P$ and $r_{i}$ is a read
with prefix $S$.
The $S\$$-interval and the $PS\$$-interval are nonempty.
Moreover $R^{p}(S)\neq\emptyset$ and $S$ a seed, hence there is an iteration of
ExtendEncodings where we backward extend the $PS\$$-encoding.
Since $r_{j}=PS$, $PS$ has a nonempty backward $\$$-extension, hence the
condition at line~\ref{alg:EAI-nonemtpy-P} is satisfied.
\end{proof}

\subsection{Extending arc labels}
% \subsection{Backward extending a generic set of string-intervals.}
% \label{sec:backward-extending-q-interval}

While the procedure \textbf{ExtendEncodings} backward extends the
$PS\$$-intervals, the actual arc labels are the $P$-intervals, therefore we need a
dedicated procedure, called \textbf{CompleteExtensions}, that scans the results of
\textbf{ExtendEncodings}, \ie a list of partially extended encodings and updates them by
extending the intervals $q(P)$ on $R$ and $q^{r}(P^{r})$ on $R^{r}$.

A procedure
\textbf{ExtendIntervals} has been originally described~\cite{bauer_lightweight_2013}
to compute all backward extensions of a set of \emph{disjoint} string-intervals, with only a single pass over $\FileBWT$.
In our case, the string-intervals are not necessarily disjoint,
therefore that procedure is not directly applicable.
We exploit the property that any two string-intervals are either nested
or disjoint to design a new procedure that computes all backward extensions
with a single scan of the list $\FileBWT$.

Our procedure \textbf{CompleteExtensions} takes in input a
list $I$ of partially extended encodings $([b,e), q(S\$), q(P), q^{r}(P^{r}),
l_, \l_{PS})$, sorted by increasing values of $b$
and (as a secondary criterion) by decreasing values of $e$.
Moreover the list $I$ is terminated with a sentinel partially extended encoding
$(\cdot, \cdot, [n+1, n+2),  \cdot, \cdot, \cdot)$---we recall that $n$ is the total
number of input characters.
For each input encoding coming from the list $\FileP(\sigma, \cdot, \cdot)$, the procedure
outputs the record $(p(PS\$), q(S\$), q(\sigma P), q^{r}(P^{r}\sigma),
l_P, \l_{PS})$.

\begin{algorithm}[tb!]
\SetKwInOut{PRE}{Input}%
\SetKwInOut{POST}{Output}
\newcommand{\Lmax}{\ensuremath{L_\mathrm{max}}}
\PRE{%
The BWT $B$ of a set $R$ of strings.
The correct lists $\FileP(\cdot, \cdot, \cdot)$ containing all
partially extended $(\sigma P, S)$-encoding such that $|P| = \l_{P}$.
}
\POST{%
The correct lists $\FileE(\cdot, \cdot, \cdot)$ containing all
$(\sigma P, S)$-encodings.
}
$I\gets$ Merge$(\{ \FileP(\sigma, \l_{PS} - \l_{P}, \l_{PS} + 1) :
\sigma\in\Sigma, \l_{PS} \ge \l_{P}\})$\;
Append the sentinel interval
$(\cdot, \cdot, [n+1, n+2), \cdot, \cdot, \cdot)$ to $I$\;
$\Pi\gets$ $|\Sigma|$-long vector $\bar{0}$\;\label{alg:EAL:init}
$Z\gets$ stack with the record $\langle (\cdot, \cdot, [1, \infty), \cdot,
\cdot, \cdot),  \Pi\rangle$\label{alg:EAL:initZ}\;
$p\gets 1$; $e_{z}\gets +\infty$\label{alg:EAL:preamble}\;
\ForEach{$(q_{1}, q_{2}, [b,e), [b', e^{r}), \l_{P}, \l_{PS}) \in I$}{\label{alg:EAL:close}\label{alg:EAL:intervals}
  $\langle (\cdot, \cdot, [\cdot,e_{z}), \cdot, \cdot, \cdot), \cdot\rangle\gets$ top$(Z)$\;
  \While{$e>e_{z}$\label{alg:EAL:begin-output}}{%
    \While{$p < e_{z}-1$\label{alg:EAL:Pi1}}{%
      $\Pi[B[p]] \gets \Pi[B[pi]] + 1$\label{alg:EAL:pi}
      $p\gets p+1$\;
    }
    $p\gets e_{z}$\label{alg:EAL:Pi-for-e}\;
    $\langle (q_{1z}=[b_{ps}, e_{ps}), q_{2z},  [b_{z}, e_{z}), [b'_z, e^{r}_z), \l_{pz}, \l_{psz})  , \Pi_{z}\rangle\gets$ pop$(Z)$\;
    $\sigma\gets C^{-1}(b_{ps})$\label{alg:EAL:extbegin}\;
      $prev \gets \sum_{c < \sigma} \left ( \Pi(c) - \Pi_{z}(c) \right)$\;
      $w \gets \Pi(\sigma) - \Pi_{z}(\sigma)$\;
        Append $(q_{1z}, q_{2z},
        [ C[\sigma] + \Pi_{z}[\sigma] + 1, C[\sigma] + \Pi_{z}[\sigma] + 1 + w),
        [b'_{z}+ prev, b'_{z}+ prev + w ),
        \l_{pz}, \l_{ez})$ to the list $\FileE(\sigma, \l_{psz}-\l_{pz}, \l_{psz})$\;\label{alg:EAL:extend}
      \label{alg:EAL:end-output}
  }
  \While{$p<b$\label{alg:EAL:begin-scan}}{%
    $\Pi[B[p]] \gets \Pi[B[p]] + 1$\label{alg:EAL:pi-2}\;
    $p\gets p+1$\label{alg:EAL:end-scan}\;
  }
  push$(Z, \langle (q_{1}, q_{2}, [b,e), [b', e^{r}), \l_{P}, \l_{PS}), \Pi\rangle)$\label{alg:EAL:pushZ}\label{alg:EAL:end}\;
\label{alg:EAL:bottom}
}
\caption{CompleteExtensions$(\l_{P})$}
\label{alg:Extend-Q-Intervals}
\end{algorithm}

Just as the procedure \textbf{ExtendEncodings}, we maintain an array $\Pi$,
where $\Pi[\sigma]$ is equal to the number of occurrences of the character $\sigma$ in
$\FileBWT[:p-1]$, and $p$ is the current position in $\FileBWT$.
The only difference \wrt \textbf{ExtendEncodings} is that $\sigma$ can be the
sentinel character $\$$.
We recall that $\Occ(\sigma, p)$ is the number of occurrences of $\sigma$ in
$B[:p-1]$~\cite{Ferragina2005}, therefore $\Pi[\sigma]=\Occ(\sigma,
p)$, where $p$ is a the number of symbols of $\FileBWT$ that have been read.
The array $\Pi$ is used to compute the backward extension at
line~\ref{alg:EAL:extend} of Algorithm~\ref{alg:Extend-Q-Intervals}.

We also maintain a stack $Z$ storing the partially extended encodings that have already been
read, but have not been extended yet.
A correct management of $Z$ allows to have the encodings in the correct order, that
is to read the encodings in increasing order of $b$, and to actually extend the
encodings in increasing order of $e$.
This ordering allows to scan sequentially $\FileBWT$.

We will start by showing that \textbf{CompleteExtensions} correctly manages the
array $\Pi$.

\begin{lemma}
\label{lemma:EAL-Pi-correct}
Assume that the input partial encodings are $([b,e) = q(\sigma PS\$), q(\$ S), q(P),
q^{r}(P^{r}), l_P, \l_{PS})$ and that they are sorted by increasing
value of $e$.
Then
after lines~\ref{alg:EAL:Pi-for-e} and~\ref{alg:EAL:end-scan}, and before
lines~\ref{alg:EAL:Pi1} and~\ref{alg:EAL:begin-scan} of
Algorithm~\ref{alg:Extend-Q-Intervals}, $\Pi[\sigma] = Occ(\sigma,p)$.
\end{lemma}

\begin{proof}
We can prove the lemma by induction on the number $i$ of input encodings that
have been read so far.
Notice that the lemma holds at line~\ref{alg:EAL:preamble}, since
$\Pi = \bar{0}$ and $p=1$.
When $i=1$, the condition at line~\ref{alg:EAL:begin-output} does not
hold, hence we only need to consider the value of $\Pi$ at
lines~\ref{alg:EAL:begin-scan} and~\ref{alg:EAL:end-scan}.
A direct inspection of lines~\ref{alg:EAL:begin-scan}--\ref{alg:EAL:end-scan}
shows that the lemma holds in this case.

Assume now that $i>1$.
Since the procedure modifies $\Pi$ or $p$ only at
lines~\ref{alg:EAL:Pi1}--\ref{alg:EAL:Pi-for-e} and
lines~\ref{alg:EAL:begin-scan}--\ref{alg:EAL:end-scan}, by inductive hypothesis
the lemma holds before line~\ref{alg:EAL:Pi1}.
Again, a direct inspection of lines~\ref{alg:EAL:Pi1}--\ref{alg:EAL:Pi-for-e}
proves that the lemma holds at line~\ref{alg:EAL:Pi-for-e}, which in turn
implies that it holds at line~\ref{alg:EAL:begin-scan}.
The same direct inspection of
lines~\ref{alg:EAL:begin-scan}--\ref{alg:EAL:end-scan} as before is sufficient
to complete the proof.
\end{proof}

\begin{algorithm}[tb!]
\SetKwInOut{PRE}{Input}
\SetKwInOut{POST}{Output}
\PRE{
A set $R$ of reads.
}
\POST{
The overlap graph $G_{O}$ of $R$.
}
$\l_{max}\gets $ the maximum length of a read in $R$\;
Construct \FileBWT, \FileLCP, \FileSA\;
BuildBasicArcIntervals$(R)$\tcc*{computes $\FileE(\cdot, \l_{PS}, \l_{PS})$}
\For{$\l_{P}\gets 0$ \KwTo $\l_{max} - 2$}{%
  ExtendEncodings($\l_{P}$)\tcc*{computes $\FileP(\cdot, \l_{PS} - \l_{p}, \l_{PS})$}
  CompleteExtensions($\l_{P}$)\tcc*{computes $\FileE(\cdot, \l_{PS} - \l_{P}, \l_{PS})$}
}
\caption{OverlapGraph$(R)$}
\label{alg:OG}
\end{algorithm}

Since elements are pushed on the stack $Z$ only at lines~\ref{alg:EAL:initZ}
and~\ref{alg:EAL:pushZ}, and  partially extended encodings are pushed on $Z$ without
any change,
a direct consequence of Lemma~\ref{lemma:EAL-Pi-correct}
is the following corollary.

\begin{corollary}
\label{corollary:Pi-in-Z}
Let $\langle (\cdot, \cdot, [b,e), \cdot, \cdot, \cdot) , \Pi\rangle$ be an
element of $Z$, and let $\sigma$ be any character in $\extalphabet$.
Then $\Pi[\sigma] = Occ(\sigma, b)$.
\end{corollary}

\begin{lemma}
\label{lemma:EAL-extension-correct}
Let $\FileP(\cdot, \cdot, \l_{PS})$ be the correct lists that are
the input of Algorithm~\ref{alg:Extend-Q-Intervals}, and let
$(q(\sigma PS\$), q(\$ S), q(P),
q^{r}(P^{r}), l_P, \l_{PS})$ be a generic partially extended encoding
in one of such lists.
Then Algorithm~\ref{alg:Extend-Q-Intervals}
outputs the encoding $(q(\sigma PS\$), q(S\$), q(\sigma P),
q^{r}(P^{r}\sigma ), \l_P, \l_{PS})$
if and only if
$(q(\sigma PS\$) = [b_{ps}, e_{ps}), q(S\$), q(P), q^{r}(P^{r}), \l_P,
\l_{PS})$ is an input partial encoding.
\end{lemma}

\begin{proof}
In Algorithm~\ref{alg:Extend-Q-Intervals} each input partially
extended encoding is pushed on the stack $Z$ exactly once and extended
exactly once.
Hence we only need to prove that for each partially extended input
encoding $(q(\sigma PS\$), q(S\$), q(\sigma P),
q^{r}(P^{r}\sigma ), \l_P, \l_{PS})$, the $(\sigma P, S)$-encoding is output.

Let $[b,e)$ be equal to $q(P)$, let $[b',e^{r})$ be equal to $q^{r}(P^{r})$ and
let $\sigma$ be the symbol $C^{-1}(b_{ps})$.
Notice that the output encoding is obtained by computing $q(\sigma P)$
and $q^{r}(P^{r}\sigma )$.
By Proposition~\ref{prop:extension-formulae}, $q(\sigma P)$
is equal to $[C[\sigma] + Occ(\sigma, b) + 1, C[\sigma] + Occ(\sigma, e) )$,
which is correctly computed at line~\ref{alg:EAL:extend} (by
Lemma~\ref{lemma:EAL-Pi-correct} and Corollary~\ref{corollary:Pi-in-Z}).

Moreover,  $q^{r}(P^{r}\sigma )$ is equal to  $[b'_{1},e^{r}_{1})$, by
Proposition~\ref{prop:extension-formulae},
Lemma~\ref{lemma:EAL-Pi-correct} and Corollary~\ref{corollary:Pi-in-Z}.
Moreover, the encoding $(q(PS\$), q(S\$), [b_{1},e_{1}), [b'_{1},e^{r}_{1}),
\l_P, \l_{PS})$ is output at line~\ref{alg:EAL:extend}.
\end{proof}

Notice that the value of $p$ never decreases, since its value is modified only
by increments.
This fact implies that $\FileBWT$ is scanned sequentially.
To complete the correctness of our algorithm, we need to show that the output
follows the desired ordering.
We will start with some lemmas showing the structure of the encodings stored in
the stack $Z$.

\begin{lemma}
\label{lemma:EAL-stack-nested}
The stack $Z$
of Algorithm~\ref{alg:Extend-Q-Intervals} contains a hierarchy
of encodings $(\cdot, \cdot, [b,e), \cdot, \cdot, \cdot)$
where all intervals
$[b,e)$ are nested, with the smallest at the top.
\end{lemma}

\begin{proof}
We only have to prove that the lemma holds at line~\ref{alg:EAL:pushZ}, since it
is the only line where an encoding is pushed on a nonempty stack.
Let $Z$ be the stack just before the push, and let $(\cdot, \cdot,
[b_{z},e_{z}), \cdot, \cdot, \cdot)$ be the encoding at the top of $Z$.

Clearly the lemma holds when $Z$ contains only the sentinel encoding pushed at
line~\ref{alg:EAL:initZ}, therefore assume that the top encoding of $Z$ is an
input encoding.

Notice that an encoding is pushed on $Z$ without modification, before reading
the next input encoding.
Since the input encodings are sorted by increasing values of $b$, then $b\ge b_{z}$.
To reach line~\ref{alg:EAL:pushZ}, the condition at
line~\ref{alg:EAL:begin-output} must be false, hence $e\le e_{z}$.
Consequently $[b,e)$ is included in $[b_{z}, e_{z})$.
\end{proof}

\begin{lemma}
\label{lemma:correct-EAL-extension}\label{lemma:extension-is-correct}
Algorithm~\ref{alg:Extend-Q-Intervals} pops all input encodings
$(\cdot, \cdot, [b,e), \cdot, \cdot, \cdot)$ in nondecreasing order of $e$, but it
does not pop the sentinel encodings.
\end{lemma}

\begin{proof}
First, we will consider a single generic iteration of the while loop at lines~\ref{alg:EAL:begin-output}--\ref{alg:EAL:extend}.
By Lemma~\ref{lemma:EAL-stack-nested} the intervals in  $Z$ are nested,
therefore the intervals popped in a single iteration satisfy the lemma.

We can consider the intervals popped in different iterations.
Let $f_{1}=(\cdot, \cdot, [b_{1},e_{1}), \cdot, \cdot, \cdot)$ be the most recently popped
encoding, and let $f_{z}=(\cdot, \cdot, [b_{z},e_{z}), \cdot, \cdot, \cdot)$ be a generic
interval that has been popped
from $Z$ in a previous iteration; we will show that $e_{z}\le e_{1}$.
Moreover let $f=(\cdot, \cdot, [b,e), \cdot, \cdot, \cdot)$ be the
encoding read from $I$ at the iteration when  $f_{1}$ has been popped from $Z$.
By construction, $f_{z}$ precedes $f$ (which precedes $f_{1}$) in $I$.
Since the intervals in $I$ are in non-decreasing order of the start boundary,
$b_{z}\le b\le b_{1}$.
Moreover, the condition at line~\ref{alg:EAL:begin-output} that determines
when to pop an encoding, implies that $e>e_{z}$.
All string-intervals in $I$ are disjoint or nested, therefore $e \le
b_{1}$ or $e\ge e_{1}$.
If $e\le b_{1}$, then $e\le e_{1}$ and, a fortiori since $e_{z}<e$, then $e_{z}\le
e_{1}$.
Hence we only need to consider the case when $[b,e)$ includes $[b_{1}, e_{1})$,
that is $e \ge e_{1}$.
Now, let us consider the intervals $[b,e)$ and $[b_{z}, e_{z})$.
Since $e_{z}<e$ and $b_{z}\le b$, those intervals cannot be nested, hence
$e_{z}\le b$.
Since $b\le b_{1}$ and $b_{1}<e_{1}$, then $e_{z}\le e_{1}$.

Finally, we want to prove that all intervals in $I$, except for the sentinel
intervals, are popped from $Z$ (and backward extended).
Just after reading from $I$ the sentinel $(\cdot, \cdot, [n+1, n+2),
\cdot, \cdot, \cdot)$, all intervals in $Z$, but not the starting sentinel,
satisfy the condition at line~\ref{alg:EAL:begin-output}, completing the proof.
\end{proof}

\begin{corollary}
\label{corollary:extensions-correct}
The lists $\FileE(\sigma, \l_{P}, \l_{PS})$ are correct.
\end{corollary}

\begin{proof}
It is a direct consequence of
Corollary~\ref{corollary:partial-extensions-correct} and Lemma~\ref{lemma:correct-EAL-extension}.
\end{proof}

\begin{corollary}
\label{corollary:alg:EAI:correct}
Algorithm~\ref{alg:OG} correctly computes the arcs of the overlap graph $G_{O}$.
\end{corollary}

\begin{proof}
It is a direct consequence of
Corollary~\ref{corollary:basic-encodings-correct} (which shows the
correctness of Algorithm~\ref{alg:BAI} to compute the basic
encodings), Corollary~\ref{corollary:partial-extensions-correct}  (which shows the
correctness of Algorithm~\ref{alg:EAI} to compute the partial
extensions of a set of encodings), Lemma~\ref{lemma:terminal-arc}
(which shows the correctness of Algorithm~\ref{alg:Extend-Q-Intervals} to complete
the extension of a set of partially extended encodings), and finally
Corollary~\ref{corollary:IO-records} which shows that the arcs of the
overlap graph $G_{O}$ are correctly output.
\end{proof}

Algorithm~\ref{alg:BAI} scans only once the lists \FileBWT, \FileSA,
\FileLCP (hence reading $3n$ records, with $n=|\FileBWT|$)
and outputs at most $n$ records, by Lemma~\ref{lemma:1-seed-opening-position}.

Each execution of
Algorithm~\ref{alg:EAI} scans only once the lists \FileBWT, \FileSA,
\FileLCP (hence reading $3n$ records)
as well as the lists $\FileE(\cdot,\cdot,\l_{PS})$ containing the
$(P,S)$-encodings with $|PS| = \l_{PS}$.
Let us now consider the $(P,S)$-encodings that are read during a
single execution of Algorithm~\ref{alg:EAI}, and notice that the
corresponding $(P,S)$-intervals are disjoint, since the length of
$|PS|$ is always equal to $\l_{PS}$.
This fact implies that at most $n$ $(P,S)$-encodings are read and at
most $n$ $(P,S)$-encodings are output.

The analysis of Algorithm~\ref{alg:Extend-Q-Intervals}  is similar to
that of Algorithm~\ref{alg:EAI}.
The only difference is that
Algorithm~\ref{alg:Extend-Q-Intervals} scans once the list \FileBWT, as well as
the lists $\FileE(\cdot,\cdot,\l_{PS})$.
The consequence is that Algorithm~\ref{alg:Extend-Q-Intervals}  reads at most $2n$ records
and writes at most $n$ records.

Since Algorithms~\ref{alg:EAI}
and~Algorithm~\ref{alg:Extend-Q-Intervals} are called at most $l$
times, the overall number of records that are read is at most $3n + 6ln$.
Notice that this I/O complexity matches the one of
BCRext~\cite{bauer_lightweight_2013}, which is the most-efficient
known external-memory algorithm to compute the data structures (GSA,
BWT, LCP) we use to index the input reads.

\begin{corollary}
\label{corollary:IO-records}
Given the lists \FileBWT, \FileSA, \FileLCP, it is possible to compute the
overlap graph of a set of reads $R$ with total length $n$ and where no read is
longer than $l$ characters, reading sequentially $(3+ 6l)n$ records.
\end{corollary}

\begin{comment}
Before extending
our analysis on the number of records that are read or written to the
actual size of data, we need some observations to improve the algorithm.
It is immediate to notice that ExtendEncodings modifies only the parts of an
encoding relative to $PS\$$ and $|P|$, \ie the first and the sixth components.
On the other hand, CompleteExtensions modifies only the parts of an
encoding relative to $P$ and $P^{r}$, \ie the third and the fourth components.
This fact suggests to split the encodings in two parts, where each part is
updated separately.
More precisely, Algorithm~\ref{alg:EAI} reads and writes only the parts
$q(PS\$)$, $q(S\$)$, $\l_{P}$, $\l_{PS}$, while
Algorithm~\ref{alg:Extend-Q-Intervals} reads only $q(PS\$)$, $q(P)$ and
$q^{r}(P^{r})$ and writes only $q(P)$ and $q^{r}(P^{r})$.
Since each string-interval is represented with $2$ integers, it is possible to
compute the overlap graph (whose representation requires at most $5$ integers
for each arc) reading $(12l^{2} + 9)n$ integers and writing $(10l^{2} + 9)n$ integers.
\end{comment}

\section{Reducing the overlap graph to a string graph}

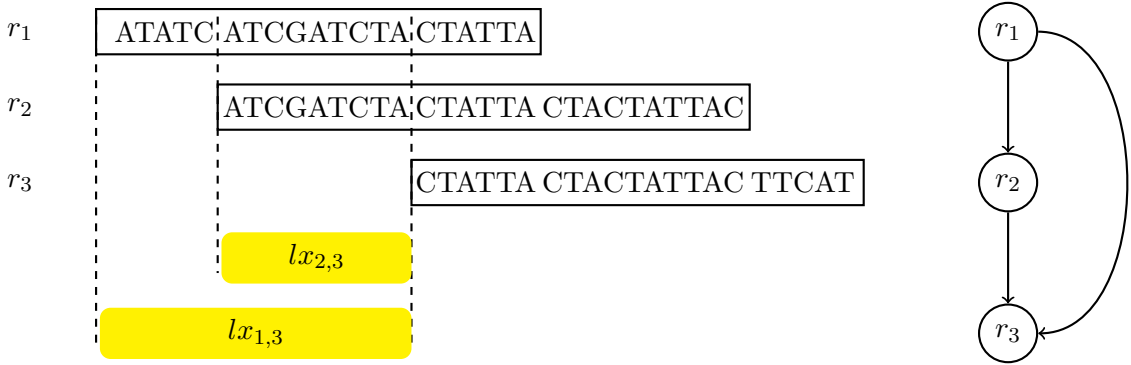
\begin{figure}[htb]
\centering
\begin{tikzpicture}[thick]
\node (r1) at (-1,2) {$r_{1}$};
\node (r2) at (-1,1) {$r_{2}$};
\node (r3) at (-1,0) {$r_{3}$};
\node (r11) at (0.9,2) {ATATC};
\node (r12) at (2.9,2) {ATCGATCTA};
\node (r13) at (5.0,2) {CTATTA};
\node (r21) at (2.9,1) {ATCGATCTA};
\node (r22) at (5.0,1) {CTATTA};
\node (r23) at (7.2,1) {CTACTATTAC};
\node (r31) at (5.0,0) {CTATTA};
\node (r32) at (7.2,0) {CTACTATTAC};
\node (r33) at (9.3,0) {TTCAT};
\draw (-0.0, 2.3) rectangle (5.85,1.7);
\draw (1.6, 1.3) rectangle (8.6,0.7);
\draw (4.15, 0.3) rectangle (10.1,-0.3);
\draw (0.0,2.2) edge [dashed] (0.0,-2.2);
\draw (1.6,2.2) edge [dashed] (1.6,-1.2);
\draw (4.15,2.2) edge [dashed] (4.15,-2.2);
% \draw (5.85,1.9) edge [dashed] (5.85,-0.3);
% \draw (8.6,1.1) edge [dashed] (8.6,-0.3);

\node (pe23) at (2.9,-1) [text width=2.20cm,rounded
corners,fill=yellow,inner sep=0.15cm, align=center]  {$lx_{2,3}$};
\node (pe13) at (2.1,-2) [text width=3.80cm,rounded
corners,fill=yellow,inner sep=0.15cm, align=center]  {$lx_{1,3}$};

\node (rr1) at (12,2) [shape=circle,draw] {$r_{1}$};
\node (rr2) at (12,0) [shape=circle,draw] {$r_{2}$};
\node (rr3) at (12,-2) [shape=circle,draw] {$r_{3}$};
\draw [->] (rr1) -- (rr2);
\draw [->] (rr2) -- (rr3);
\draw [->] (rr1) to [out=0,in=0] (rr3);
\end{tikzpicture}
\caption{Example of reducible arc of the overlap graph.
The read $r_{1}$ is equal to ATATCATCGATCTACTATTA, while the read $r_{2}$ is
equal to ATCGATCTACTATTACTACTATTAC and the read $r_{3}$ is CTATTACTACTATTACTTCAT.
The associated overlap graph is on the right.
The arc   $(r_{1},r_{3})$   is reducible.
}
\label{fig:arc-reduce}
\end{figure}

In this section we state a  characterization of string graphs based on
the notion of string-interval, then we will exploit such
characterization to reduce the overlap graph.

\begin{lemma}
\label{lemma:reducible-if-prefix-interval-is-suffix}
Let $G_{O}$ be the labeled overlap graph for a substring-free set $R$ of reads and let
$(r_{i},r_{j})$ be an arc of $G_{O}$.
Then, $(r_{i},r_{j})$ is reducible iff there exists another
arc $(r_{h},r_{j})$ of $G_{O}$ incoming in $r_{j}$ and such that
$lx^{\rev}_{h,j}$ is a proper prefix of $lx^{\rev}_{i,j}$.
\end{lemma}

\begin{proof}
First notice that $lx^{\rev}_{h,j}$ is a proper prefix of $lx^{\rev}_{i,j}$ iff
and only if $lx_{h,j}$ is a proper suffix of $lx_{i,j}$.

By definition, $(r_{i},r_{j})$ is reducible if and only if there exists a second
path $r_{i}, r_{h_{1}}, \ldots , r_{h_{k}}, r_{j}$ representing the
string $XYZ$, where $X$, $Y$ and $Z$ are respectively the left
extension of $r_{j}$ with $r_{i}$, the
overlap of $r_{i}$ and $r_{j}$, and the right extension of $r_{i}$ with $r_{j}$.

Assume that such a path $(r_{i}, r_{h_{1}}, \ldots , r_{h_{k}}, r_{j})$
exists.
Since $r_{i}, r_{h_{1}}, \ldots , r_{h_{k}}, r_{j}$ represents $XYZ$ and
$Z=rx_{i,j}$, $r_{h_{k}}=X_{1}YZ_{1}$ where $X_{1}$ is a suffix of $X$ and
$Z_{1}$ is a proper prefix of $Z$. Notice that $X_{1}=lx_{h_{k}, j}$ and $R$ is
substring free, hence $X_{1}$ is a proper suffix of $X$, otherwise $r_{i}$ would
be a substring of $r_{h_{k}}$, completing this direction of the proof.

Assume now that there exists an arc $(r_{h},r_{j})$ such that $lx_{h,j}$ is a
proper suffix of $lx_{i,j}$. Again, $r_{h}=X_{1}Y_{1}Z_{1}$ where $X_{1}$,
$Y_{1}$ and $Z_{1}$ are respectively the left
extension of $r_{j}$ with $r_{h}$, the
overlap of $r_{h}$ and $r_{j}$, and the right extension of $r_{h}$ with $r_{j}$.
By hypothesis, $X_{1}$ is a proper suffix of
$X$. Since $r_{h}$ is not a substring of $r_{i}$, the fact that $X_{1}$ is a
suffix of $X$ implies that $Y$ is a substring of $Y_{1}$, therefore $r_{i}$ and
$r_{h}$ overlap and $|ov_{i, h}|\ge |Y|$, hence $(r_{i}, r_{h})$ is an
arc of $G_{O}$.

The string associated to the path $r_{i}, r_{h}, r_{j}$ is
$r_{i} rx_{i, h} rx_{h, j}$. By Lemma~\ref{lemma:prefix-interval-path},
$r_{i} rx_{i, h} rx_{h, j}= lx_{i, h} lx_{h, j} r_{j}$. At the same time the
string associated to the path $r_{i}, r_{j}$ is
$r_{i} rx_{i, j} = lx_{i, j} r_{j}$ by Lemma~\ref{lemma:prefix-interval-path},
hence it suffices to prove that $lx_{i, h} lx_{h, j}=lx_{i, j}$. Since
$lx_{h,j}$ is a proper suffix of $lx_{i,j}$, by definition of left extension,
$lx_{i, h} lx_{h, j}=lx_{i, j}$, completing the proof.
\end{proof}

Since the
encoding of an arc $(r_{h}, r_{j})$ contains both $q(lx^{\rev}_{h,j})$ and
$|lx^{\rev}_{h,j}|$, we can transform
Lemma~\ref{lemma:reducible-if-prefix-interval-is-suffix}
into an easily testable property, by way of Proposition~\ref{proposition:test-prefix}.
The following Lemma~\ref{lemma:string-graph-arc-reduces} shows that, if $(r_{x},
r_{z})$ can be reduced, then it can be reduced by an arc of the string
graph $G_{S}$, hence avoiding a comparison between all pairs of arcs
of $G_{O}$.

\begin{lemma}
\label{lemma:string-graph-arc-reduces}
Let $G_{O}$ be the overlap graph of a string $R$ of reads, let $G_{S}$
be the corresponding string graph,
and let  $(r_{x}, r_{z})$ be an arc of $G_{O}$ that is not an arc of $G_{S}$.
Then there exists an arc $(r_{s}, r_{z})$ of $G_{S}$ such that
$q^{r}(lx^{r}_{s,z})$ includes
$q^{r}(lx^{r}_{x,z})$ and $|lx_{x,z}| > |lx_{s,z}|$.
\end{lemma}

\begin{proof}
Let $(r_{s}, r_{z})$ be the arc of $G_{O}$ whose left extension is the
shortest among all arcs of $G_{O}$ such that $q^{r}(lx^{r}_{s,z})$ includes
$q^{r}(lx^{r}_{x,z})$ and $|lx_{x,z}| > |lx_{s,z}|$.
By  Lemma~\ref{lemma:reducible-if-prefix-interval-is-suffix}, since
$(r_{x}, r_{z})$  is not an arc of $G_{S}$ such an arc must exist.
We want to prove that  $(r_{s}, r_{z})$  is an arc of $G_{S}$.

Assume to the contrary that $(r_{s}, r_{z})$  is not an arc of
$G_{S}$, that is there exists an arc $e_{1}=(r_{h}, r_{z})$ of $G_{O}$ such that
$q^{r}(lx_{s,z})$ includes $q^{r}(lx_{h,z})$ and $|lx_{s,z}| >
|lx_{h,j}|$.
Then $q^{r}(lx_{x,z})$ includes $q^{r}(lx_{h,z})$ and $|lx_{x,z}| >
|lx_{h,j}|$, contradicting the assumption that  $(r_{s}, r_{z})$ is
the arc of $G_{O}$ whose left extension is the
shortest among all arcs of $G_{O}$ such that $q^{r}(lx^{r}_{s,z})$ includes
$q^{r}(lx^{r}_{x,z})$ and $|lx_{x,z}| > |lx_{s,z}|$.
\end{proof}

Lemma~\ref{lemma:string-graph-arc-reduces} suggests that each arc $e$ of
$G_{O}$ should be tested only against arcs in $G_{S}$ whose left
extension is strictly shorter than that of $e$ to determine
whether $e$ is also an arc of $G_{S}$.
A simple comparison between each arc of the overlap graph and each arc
of the string graph would determine which arcs are irreducible, but
this approach would require to store in main memory all arcs of
$G_{S}$ incident on a vertex.
To reduce the main memory usage, we partition the arcs of $G_{S}$
incoming in a vertex $r_{z}$ into chunks, where each chunk can contain at most
$M$ arcs (for any given constant $M$)~\cite{Vitter:2001:EMA:384192.384193}.
Let $D_{z}$ be the set of arcs of $G_{S}$ incoming in $r_{z}$, and let
$d_{z}$ be its cardinality.
Since there are at most $M$ arcs in each chunk, we need $\lceil
d_{z}/M\rceil$ passes over $D_{z}$ to perform all comparisons.
There are some technical details that are due to the fact that the set
$D_{z}$ is not known before examining the arcs of $G_{O}$ incoming in
$r_{z}$ (see Algorithm~\ref{alg:reduce-graph-RG}).
Mainly, we need an auxiliary file to store whether each arc $e$ of $G_{O}$
has already been processed, that is if we have already decided
whether $e$ is an arc of $G_{S}$.

\begin{algorithm}[tb!]
\SetKwInOut{PRE}{Input}%
\SetKwInOut{POST}{Output}
\PRE{%
The number $M$ of $5$-integer records that can be stored in memory.
The lists $\FileArc(p, \cdot)$  of arc encodings of $G_{O}$.
}
\POST{The set $E$ of arc encodings of the irreducible arcs of $G_{O}$.}
$\l_{\max}$ the maximum length of a read in $R$\;
$|R|$ the number of reads in  $R$\;
\For{$z\gets 1$ to $|R|$\label{alg:RG:outer}}{%
  $D_{z}\gets \emptyset$\;
  \While{$\exists p: \FileArc(p,z)$ contains at least an encoding not marked processed\label{alg:RG:main-while}}{%
    $C\gets \emptyset$\tcc*{$C$ contains a chunk of edges of $G_{S}$}
    \For{$p\gets 1$ to $\l_{\max}$}{%
      \ForEach{unprocessed arc encoding $e=\langle i, z, q^{r}(Q^{r}), |Q|
        \rangle \in \FileArc(p,z)$\label{alg:RG:begin-1}}{\label{alg:reduce:foreach-1edge:start}%
        \ForEach{$\langle h, z, q^{r}(P^{r}), |P|\rangle \in C$}{\label{alg:reduce:foreach:start}%
          \If{$|P| < |Q|$ and $q^{r}(P^{r})$ contains
            $q^{r}(P^{r})$\label{alg:RG:test}}{%
            Mark $e$ as transitive and
            processed\;\label{alg:reduce:foreach:end}
          }
        }
        \If{$|C|<M$ and $e$ is not marked transitive}{%
          Add $e$ to $C$\;
          Mark $e$ as processed\;
        }\label{alg:reduce:foreach-1edge:end}
      }
      $D_{z} \gets D_{z} \cup C$\;\label{alg:RG:end-1}\label{alg:RG:main-while:end}
    }
  }
}
\Return $\cup_{z}D_{z}$\;
\caption{ReduceOverlapGraph($M$)}
\label{alg:reduce-graph-RG}
\end{algorithm}

Now we can start proving the the correctness of
Algorithm~\ref{alg:reduce-graph-RG}.
\begin{lemma}
\label{lemma:reduce-terminates}
Let $G_{O}$ be the overlap graph of a string $R$ of reads, and let
$G_{S}$ be the corresponding string graph.
Then the execution of Algorithm~\ref{alg:reduce-graph-RG} on $G_{O}$
terminates with all arcs of $G_{O}$ marked processed.
\end{lemma}

\begin{proof}
To prove that the algorithm terminates, we only have to prove that all
arcs in the generic set $\FileArc(p,z)$ are marked processed, as in
that case the condition of the while at line~\ref{alg:RG:main-while}
becomes false.
As long as there is an unprocessed arc, the condition at
line~\ref{alg:RG:main-while} is satisfied, hence the corresponding
while loop is executed.
At each execution of such loop, the first unprocessed arc is added to
$C$ (since $C$ is emptied at the beginning of the iteration) and
marked processed.
Hence, eventually all arcs must be marked processed.
\end{proof}

\begin{lemma}
\label{lemma:removed-arcs-are-transitive}
Let $G_{O}$ be the overlap graph of a string $R$ of reads,  let
$G_{S}$ be the corresponding string graph, and let
$e=\langle i, z, q^{r}(Q^{r}), |Q| \rangle $ be an arc
encoding that is marked transitive.
Then $e$ is not an arcs of $G_{S}$.
\end{lemma}

\begin{proof}
Since $e$ is marked transitive,
there exists an arc encoding $\langle h, z, q^{r}(P^{r}), |P|
\rangle \in C$ such that $|P| < |Q|$ and $q^{r}(P^{r})$ contains $q^{r}(P^{r})$.
By construction of arc encoding and by
Lemma~\ref{lemma:reducible-if-prefix-interval-is-suffix}, the arc
$(r_{i}, r_{z})$ cannot be an arc of $G_{S}$.
\end{proof}

\begin{lemma}
\label{lemma:string-graph-arcs-are-irreducible}
Let $G_{O}$ be the overlap graph of a string $R$ of reads, let
$G_{S}$ be the corresponding string graph, and let
$e=\langle i, z, q^{r}(Q^{r}), |Q| \rangle$ be an arc encoding
inserted into
$D_{z}$ by Algorithm~\ref{alg:reduce-graph-RG}.
Then $(r_{i}, r_{z})$ is an arc of $G_{S}$.
\end{lemma}

\begin{proof}
Since $e$ is in $D_{z}$, previously $e$ has been added to $C$.
Let us consider the iteration when $e$ is added to $C$: notice that
$|C|<M$ and $e$ is not marked as processed at the beginning of the iteration.
Consequently, no arc encoding that is currently in  $C$ or that has been in
$C$ in a previous iterations of the while loop at
lines~\ref{alg:RG:main-while}--\ref{alg:RG:main-while:end} satisfies
the condition of
Lemma~\ref{lemma:reducible-if-prefix-interval-is-suffix}, that is no
arc in $C$ or in a previous occurrence of $C$ can reduce $e$.

Since the arcs incoming in $r_{z}$ are examined in increasing order of
their left extension, all arcs of $G_{S}$ that are incoming in $r_{z}$
and whose left extension is shorter than $e$ have already been
inserted in $C$, either in the current iteration or in one of the
previous iterations.
Consequently no arc of $G_{O}$ can reduce $e$, hence $e$ is an arc of $G_{S}$.
\end{proof}

\begin{theorem}
\label{theorem:correct-string-graph}
Let $G_{O}$ be the overlap graph of a string $R$ of reads.
Then Algorithm~\ref{alg:reduce-graph-RG} outputs the set $E$ of the arc
encodings of the irreducible arcs of $G_{O}$ reading or writing at
most $3|E(G_{O})|\lceil d/M\rceil$ records, where $E(G_{O})$ is the
arc set of $G_{O}$ and $d$ is the maximum indegree of $G_{S}$.
\end{theorem}

\begin{proof}
By
Lemmas~\ref{lemma:reduce-terminates},~\ref{lemma:removed-arcs-are-transitive},
and~\ref{lemma:string-graph-arcs-are-irreducible}, Algorithm~\ref{alg:reduce-graph-RG} outputs the set $E$ of the arc
encodings of the irreducible arcs of $G_{O}$.

To determine the total number of records that are read by
Algorithm~\ref{alg:reduce-graph-RG}, we notice that each execution of
the while loop
at lines~\ref{alg:RG:main-while}--\ref{alg:RG:main-while:end} read the
records of all arcs incoming in $r_{z}$ as well as all records of the
auxiliary file storing whether an arc encoding has been processed.
Moreover during each iteration, such auxiliary file is written.
Therefore the  I/O complexity of an iteration of the while loop
regarding the arc incoming in $r_{z}$ is equal to $3$ times the number
of arcs of $G_{O}$ incoming in $r_{z}$.

Now we have to determine the number of iterations of the while loop at
lines~\ref{alg:RG:main-while}--\ref{alg:RG:main-while:end}.
A consequence of Lemmas~\ref{lemma:reduce-terminates},~\ref{lemma:removed-arcs-are-transitive},
and~\ref{lemma:string-graph-arcs-are-irreducible} is that the
condition at line~\ref{alg:RG:main-while} becomes false (and we exit
from the while loop) only when all arcs in $D_{z}$
are inserted in $C$ in some iteration.
The condition at line~\ref{alg:RG:test} that an arc encoding $e$ is
added to $C$ only if $|C|<M$ and $e$ is not transitive.
Therefore only the last iteration can terminate with a set $C$
containing fewer than $M$ elements.
Hence the number of iterations is equal to $\lceil |D_{z}|/M \rceil$.
Consequently the I/O complexity of an iteration of the for loop over
all reads in $R$
(lines~\ref{alg:RG:outer}--~\ref{alg:RG:main-while:end}) is equal to
$3 \lceil |D_{z}|/M \rceil |E_{O}(r_{z})|)$, where $E_{O}(r_{z})$ is the
set of arcs of $G_{O}$ that are incoming in $r_{z}$.

Summing over all iterations of the for loop at
lines~\ref{alg:RG:outer}--\ref{alg:RG:main-while:end} immediately
proves the theorem.
\end{proof}

\section{Conclusions}

The first contribution of this paper is a compact representation of the overlap
graph and of the string graph via string-intervals.
More precisely, we have shown how a string-interval can be used to represent the
set of reads sharing a common prefix, with a possible reduction in the overall
space used.

Then, we have proposed the first known external-memory algorithm to compute
the overlap graph, showing that it reads at most  $(3+ 6l)n$
records, where $n$ is the total length of
the input and $l$ is the maximum length of each input string, using
only a constant amount of main memory.
A fundamental technical contribution is the improvement of the CompleteExtensions
procedure that has been introduced in~\cite{Cox2012} to compute, with a single
scan of the BWT, all backward $\sigma$-extensions of a set of disjoint string-intervals.
Our improvement allows to extend a generic set of string-intervals.

Finally, we have described a new external-memory algorithm for reducing an
overlap graph to obtain the corresponding string graph, reading or
writing $3|E(G_{O})|\lceil d/M\rceil$ records, where $E(G_{O})$ is the arc
set of $G_{O}$ and $d$ is the maximum indegree of $G_{S}$,
while using an amount of main memory necessary
to store $M/5$ integers (as well as some constant-sized data structure).

There are some open problems that we believe are interesting.
The analysis of the algorithm complexity is not very detailed.
In fact, we conjecture that some clever organization of the records
and a more careful analysis will show that the actual I/O complexity
is better than the one we have shown in the paper.

Another direction is to assess the actual performance of the algorithm on data originating from a set of
sequences, such as those coming from transcriptomics~\cite{Beretta2013,Lacroix2008} or
metagenomics~\cite{pell_scaling_2012}, especially to verify the gain in disk usage.

\section*{Acknowledgements}

The authors acknowledge the support of the MIUR PRIN 2010-2011 grant
``Automi e Linguaggi Formali: Aspetti Matematici e Applicativi''  code 2010LYA9RH,
of the Cariplo Foundation grant 2013-0955 (Modulation of anti cancer immune
response by regulatory non-coding RNAs), of the FA 2013 grant ``Metodi
algoritmici e modelli: aspetti teorici e applicazioni in
bioinformatica'' code 2013-ATE-0281, and of the FA 2014 grant
``Algoritmi e modelli computazionali: aspetti teorici e applicazioni
nelle scienze della vita'' code 2014-ATE-0382.

\bibliographystyle{abbrv}
\bibliography{biblioBWTpaper}

\end{document}